\documentclass[journal]{IEEEtran}
\IEEEoverridecommandlockouts

\usepackage{amssymb}
\usepackage{amsmath}
\usepackage{amsthm} 
\usepackage[ruled,linesnumbered,vlined]{algorithm2e}
\usepackage{algorithmicx}
\usepackage{algpseudocode}
\usepackage{comment}
\usepackage{graphicx}
\usepackage{gnuplottex}
\usepackage{url}
\usepackage{booktabs}
\usepackage{subfigure}

\usepackage{pgf}
\usepackage{tikz}
\usetikzlibrary{automata,arrows,positioning,calc}

\makeatletter
\DeclareRobustCommand*\cal{\@fontswitch\relax\mathcal}
\makeatother

\newtheorem{definition}{Definition}
\newtheorem{lemma}{Lemma}
\newtheorem{theorem}{Theorem}
\newtheorem{problem}{Problem}
\newtheorem{remark}{Remark}

\DeclareMathOperator*{\argmin}{arg\,min}
\algdef{SE}[DOWHILE]{Do}{DoWhile}{\algorithmicdo}[1]{\algorithmicwhile\ #1}

\title{From Entanglement Purification Scheduling to Fidelity-constrained Multi-Flow Routing}

\author{Ziyue Jia, Lin Chen \\
School of Computer Science and Engineering, Sun Yat-sen University, Guangzhou, 510006, China \\
Guangdong Provincial Key Laboratory of Information Security Technology, Guangzhou 510006, China \\
Email: jiazy5@mail2.sysu.edu.cn, chenlin69@mail.sysu.edu.cn
\thanks{This work is supported in part by National Science Foundation of China under Grant 62172455, Pearl River Talent Program under Grant 2019QN01X140, and Guangdong Provincial Key Laboratory of Information Security Technology (No. 2023B1212060026). L. Chen is the corresponding author.}}

\begin{document}

\maketitle

\begin{abstract}
    Recently emerged as a disruptive networking paradigm, quantum networks rely on the mysterious quantum entanglement to teleport qubits without physically transferring quantum particles. However, the state of quantum systems is extremely fragile due to environment noise. A promising technique to combat against quantum decoherence is entanglement purification. To fully exploit its benefit, two fundamental research questions need to be answered: (1) given an entanglement path, what is the optimal entanglement purification schedule? (2) how to compute min-cost end-to-end entanglement paths subject to fidelity constraint? In this paper, we give algorithmic solutions to both questions. For the first question, we develop an optimal entanglement purification scheduling algorithm for the single-hop case and analyze the \textsc{purify-and-swap} strategy in the multi-hop case by    establishing the closed-form condition for its optimality. For the second question, we design a polynomial-time algorithm constructing an $\epsilon$-optimal fidelity-constrained path and then a randomized optimization framework for the multi-flow case. The effectiveness of our algorithms are also numerically demonstrated by extensive simulations.
\end{abstract}


\section{Introduction}

\IEEEPARstart{F}{our} decades ago in 1982, Richard Feynman postulated  in his visionary article, \textit{Simulating Physics with Computers}~\cite{1982IJTP...21..467F}, that to simulate quantum systems one would need to build quantum computers, marking the birth of quantum computing. Today, quantum computing turns 40, we are still far from seeing Feynman's dream come true. For example, solving many fundamental problems in physics and chemistry requires hundreds of thousands or millions of qubits' computation power in order to correct errors arising from noise~\cite{8093785}. Any single quantum computer today cannot fulfill such tasks. We hence need  to interconnect multiple quantum computers to scale up the number of qubits, thus forming a \textit{quantum network}, or more ambitiously, a \textit{quantum Internet}~\cite{doi:10.1126/science.aam9288}.

A quantum network interconnects quantum devices by exploiting fundamental quantum mechanical phenomena such as superposition, entanglement, and quantum measurement to exchange quantum information in the form of qubits, thereby achieve processing capabilities beyond what is possible with classical computer systems and networks~\cite{doi:10.1126/science.aam9288}. A quantum network can form a virtual quantum machine of a large number of qubits, scaling with the number of interconnected devices. This, in turn, leads to an exponential boost of computing power with just a linear amount of the physical resources, i.e., the number of connected quantum devices. Therefore, a quantum network can support many ground-breaking applications lying beyond the capability of its classical counterparts, such as quantum communication, clock synchronization, secure remote computation, and distributed consensus.

However, designing 
quantum networks is a technically challenging task that we have never encountered in computer science and engineering. This is because quantum networks follow the laws of quantum mechanics, such as
no-cloning, quantum measurement, quantum entanglement
and teleportation~\cite{nielsen_chuang_2010}. These quantum phenomena have no counterpart in classical networks, and pose highly non-trivial constraints on the network design and optimization. Therefore, designing and optimizing
quantum networks cannot be achieved by extrapolating the classical models to their quantum analogues, but requires a fundamental paradigm shift. 
\begin{figure}
\centering
\includegraphics[width=8.5cm]{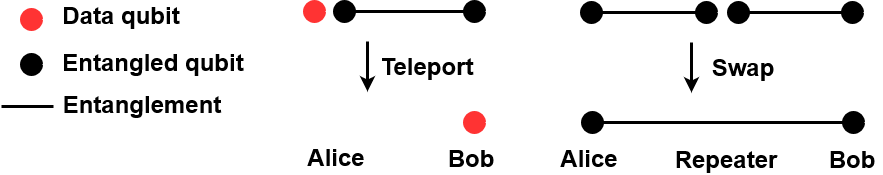}
\caption{Quantum teleportation and entanglement swapping}
\label{fig:swap}
\end{figure}

Technically, the most fundamental characteristic of quantum networks, compared to classical networks, is the famous and mysterious \textit{quantum entanglement}, allowing to realize quantum teleportation so as to transmit a qubit without physically transferring the particle storing it, as illustrated in Figure~\ref{fig:swap}.  

However, distributed qubits inevitably interact with the environment, causing \textit{decoherence}. Consequently, the initially maximally entangled quantum state will decohere into a mixed state, thus degrading or even paralyzing the corresponding quantum communication session. To combat against decoherence, there are three solutions: quantum error correction~\cite{QuantumErrorCorrection}, entanglement purification~\cite{chpt-purif}, and quantum communication based on decoherence-free subspace~\cite{10.5555/2011772.2011779}. Currently, entanglement purification appears to be the most convenient and deployable solution that has already been demonstrated both theoretically and experimentally~\cite{pan2001}.

Entanglement purification, also termed as entanglement distillation, is to extract from a pool of low-fidelity\footnote{In quantum information theory, \textit{fidelity} is the probability that a set of qubits are actually in the
state we believe that they ought to be in.} Einstein-Podolsky-Rosen (EPR) pairs a high-fidelity EPR pair\footnote{cf. Section~\ref{sec:purif} for more detailed presentation}. Though seemingly intuitive to understand, exploiting entanglement purification to establish high-fidelity end-to-end entanglement paths is by no means a trivial task due to the following three technical challenges.

\textbf{Entanglement purification scheduling}. In order to produce a high-fidelity EPR pair, we may need to perform multiple rounds of entanglement purification. Figure~\ref{fig:puri-sch} illustrates two entanglement purification scheduling policies (\textsc{symmetric} and \textsc{pumping}) resulting in different final fidelity. Their purification success probabilities are highlighted in red font. A natural question is to find the optimal entanglement purification scheduling policy.

\begin{figure}[!ht]
\centering
\includegraphics[width=0.48\textwidth]{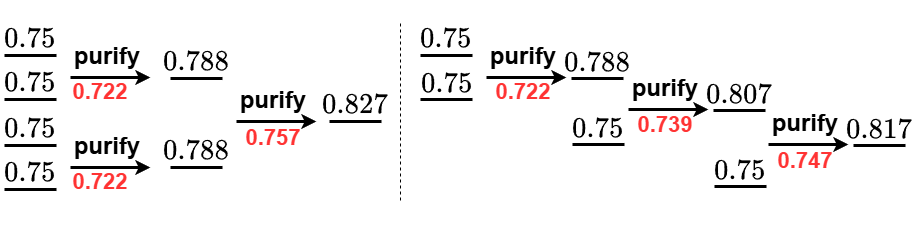}
\caption{We dispose $4$ EPR pairs of fidelity $0.75$. The left and right sub-figures depict $2$ purification scheduling policies producing $2$ final EPR pairs of different fidelities. Their purification success probabilities are highlighted in red font.}
\label{fig:puri-sch}
\end{figure}

\textbf{Interaction between entanglement purification and entanglement swapping}. Suppose that Alice and Bob are connected by a quantum repeater Charlie, with whom they have both established an EPR pair, as depicted in Figure~\ref{fig:puri-swap}. We have two different ways to produce a final EPR pair between Alice and Bob. The first one, termed as \textsc{purify-and-swap}, consists of performing entanglement purification between Alice and Charlie, and Bob and Charlie, and then entanglement swapping to stitch the two purified EPR pairs. The second one, termed as \textsc{swap-and-purify}, consists of performing entanglement swapping at Charlie to produce two EPR pairs between Alice and Bob, and then entanglement purification to get a final purified EPR pair. In fact, the above two schemes represent two extremities regarding the interaction between entanglement purification and entanglement swapping. The strategy space increases exponentially as the quantum network scales. Clearly, we face the problem of finding the optimal joint entanglement purification scheduling and swapping strategy.

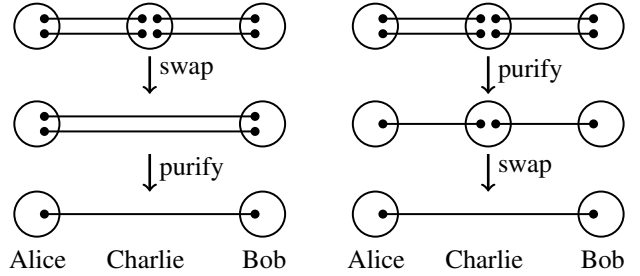
\begin{figure}[!ht]
\centering
\begin{tikzpicture}[line width=0.75pt]
\draw (0,0) circle (0.3);
\fill (0.1,0.1) circle (0.06);
\fill (0.1,-0.1) circle (0.06);
\draw (1.5,0) circle (0.3);
\fill (1.6,0.1) circle (0.06);
\fill (1.6,-0.1) circle (0.06);
\fill (1.4,0.1) circle (0.06);
\fill (1.4,-0.1) circle (0.06);
\draw (3,0) circle (0.3);
\fill (2.9,0.1) circle (0.06);
\fill (2.9,-0.1) circle (0.06);
\draw (0.1,0.1) -- (1.4,0.1);
\draw (0.1,-0.1) -- (1.4,-0.1);
\draw (1.6,0.1) -- (2.9,0.1);
\draw (1.6,-0.1) -- (2.9,-0.1);
\draw [line width=1pt,->] (1.5,-0.4) -- (1.5,-0.9);
\node[right] at (1.5,-0.6) {swap};
\draw (0,-1.3) circle (0.3);
\fill (0.1,-1.2) circle (0.06);
\fill (0.1,-1.4) circle (0.06);
\draw (3,-1.3) circle (0.3);
\fill (2.9,-1.2) circle (0.06);
\fill (2.9,-1.4) circle (0.06);
\draw (0.1,-1.2) -- (2.9,-1.2);
\draw (0.1,-1.4) -- (2.9,-1.4);
\draw [line width=1pt,->] (1.5,-1.7) -- (1.5,-2.2);
\node[right] at (1.5,-1.9) {purify};
\draw (0,-2.5) circle (0.3);
\fill (0.1,-2.5) circle (0.06);
\draw (3,-2.5) circle (0.3);
\fill (2.9,-2.5) circle (0.06);
\draw (0.1,-2.5) -- (2.9,-2.5);
\node[right] at (-0.5,-3.1) {Alice};
\node[right] at (0.8,-3.1) {Charlie};
\node[right] at (2.6,-3.1) {Bob};

\draw (4.5,0) circle (0.3);
\fill (4.6,0.1) circle (0.06);
\fill (4.6,-0.1) circle (0.06);
\draw (6,0) circle (0.3);
\fill (6.1,0.1) circle (0.06);
\fill (6.1,-0.1) circle (0.06);
\fill (5.9,0.1) circle (0.06);
\fill (5.9,-0.1) circle (0.06);
\draw (7.5,0) circle (0.3);
\fill (7.4,0.1) circle (0.06);
\fill (7.4,-0.1) circle (0.06);
\draw (4.6,0.1) -- (5.9,0.1);
\draw (4.6,-0.1) -- (5.9,-0.1);
\draw (6.1,0.1) -- (7.4,0.1);
\draw (6.1,-0.1) -- (7.4,-0.1);
\draw [line width=1pt,->] (6,-0.4) -- (6,-0.9);
\node[right] at (6,-0.6) {purify};
\draw (4.5,-1.3) circle (0.3);
\fill (4.6,-1.3) circle (0.06);
\draw (6,-1.3) circle (0.3);
\fill (5.9,-1.3) circle (0.06);
\fill (6.1,-1.3) circle (0.06);
\draw (7.5,-1.3) circle (0.3);
\fill (7.4,-1.3) circle (0.06);
\draw (4.6,-1.3) -- (5.9,-1.3);
\draw (6.1,-1.3) -- (7.4,-1.3);
\draw [line width=1pt,->] (6,-1.7) -- (6,-2.2);
\node[right] at (6,-1.9) {swap};
\draw (4.5,-2.5) circle (0.3);
\fill (4.6,-2.5) circle (0.06);
\draw (7.5,-2.5) circle (0.3);
\fill (7.4,-2.5) circle (0.06);
\draw (4.6,-2.5) -- (7.4,-2.5);
\node[right] at (4,-3.1) {Alice};
\node[right] at (5.3,-3.1) {Charlie};
\node[right] at (7.1,-3.1) {Bob};
\end{tikzpicture}
\caption{\textsc{swap-and-purify} vs. \textsc{purify-and-swap}}
\label{fig:puri-swap}
\end{figure}

\textbf{Fidelity-constrained entanglement path optimization}. At the routing level, we face the problem of finding an optimal entanglement path between Alice and Bob satisfying end-to-end fidelity constraint. By optimal, we mean to minimize the cost of the path, e.g., in terms of path delay, the total number of consumed EPR pairs. Our path optimization problem differs from its classical peer problems in that the number of qubits at each quantum nodes is limited, hence limiting the number of EPR pairs it can build with its neighbors. This non-standard constraint poses non-trivial technical challenges and thus calls for new algorithmic techniques that cannot build on existing path optimization solutions. The problem is even more challenging if we have multiple Alice-Bob pairs and we need to establish end-to-end EPR pairs with fidelity guarantee subject to the resource constraint in terms of qubits at each quantum repeater.

Driven by the above non-classic technical challenges, we embark in this paper to build a comprehensive algorithmic framework on entanglement purification scheduling and routing~\footnote{An earlier version of this paper was presented in part at the IEEE ICNP, 2024.}. Our main contributions are articulated as follows.
\begin{itemize}
    \item \textbf{Entanglement purification scheduling}. Given a pool of elementary EPR pairs, we develop an optimal entanglement purification scheduling algorithm of polynomial-time complexity. For the multi-hop case along an entanglement path, we establish the close-form condition under which \textsc{purify-and-swap} is optimal.
    \item \textbf{Fidelity-constrained entanglement path optimization}. we design a polynomial-time algorithm constructing an $\epsilon$-optimal fidelity-constrained path.
\end{itemize}

We expect that our results obtained in this paper constitutes a small but systematic step towards building non-classical entanglement provisioning and routing algorithms to support high-fidelity end-to-end quantum information transfer by exploiting the limited quantum resources in an optimal and cost-effective manner in order to fully realize the unrivallable capabilities offered by quantum networks.
\color{black}
\textcolor{black}{\textbf{Roadmap.} The paper is organized as follows. In Section~\ref{sec:purif}, we formulate and analyze the problem of entanglement purification scheduling. In Section~\ref{sec:rout}, we further analyze the problem of fidelity-constrained entanglement path optimization. Section~\ref{sec:simu} presents simulation results. We review related work in Section~\ref{sec:related-work}.
Section~\ref{sec:conclu} concludes the paper.}

\section{Entanglement Purification Scheduling}
\label{sec:purif}

As mentioned in the Introduction, the state of quantum systems is extremely fragile due to environment noise. Errors result in continuous degradation of our knowledge about the state of the quantum system. A widely applied quantitative metric about the quality of quantum states is \textit{fidelity}. Ranging from $0$ to $1$, fidelity is essentially the probability of a qubit or a set of qubits being in the state we believe that they ought to be in. To improve the fidelity of an EPR pair, \textit{entanglement purification} is usually applied. \textcolor{black}{It consumes
lower-fidelity EPR fairs to obtain a higher-fidelity one, as illustrated in Figure~\ref{fig:purif}.\footnote{Generically, purification is the process of improving our knowledge of
the state by testing propositions. This improvement is reflected as an
increase in the fidelity in the density matrix representing our knowledge about the state.}}

\begin{figure}[!ht]
\centering
\includegraphics[width=0.3\textwidth]{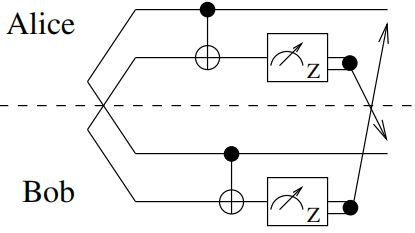}
\caption{Illustration of entanglement purification~\cite{qnetbook}. Alice and Bob each holds one half of two EPR pairs, the first pair to be purified and the second the pair to be sacrificed. The second pair is measured. The arrows indicate classical message exchange of the measurement results. If the measurement results are the same, i.e., $00$ or $11$, the purification is successful.}
\label{fig:purif}
\end{figure}

\subsection{Technical Background on Entanglement Purification}

We consider the generic scenario where entanglement purification takes as input two imperfect EPR pairs, both in the Werner state~\cite{PhysRevA.40.4277}:
\begin{align*}
    \rho_i=\frac{1-f_i}{3}I_4+\frac{4f_i-1}{3}|\Phi^+\rangle\langle\Phi^+|  \quad i=1,2,
\end{align*}
where $f_i$ denotes the fidelity of EPR pair $i$. Two \textsc{CNOT} operations are performed between the qubits at the same end of the two EPR pairs held by Alice and Bob, followed by a measurement on the two qubits of the second EPR pair in the computational basis. If the particles are found in the same state, the first EPR pair is kept, otherwise it is discarded. The fidelity of the purified EPR pair, denoted by $P(f_1,f_2)$, and the purification success probability, denoted by $F(f_1,f_2)$, can be computed in~\eqref{eq:purif} and~\eqref{eq:purif-p}, respectively~\cite{7010905}. If multiple rounds of entanglement purification are performed, we need to turn the purified EPR pair to the Werner state by executing a depolarization operation.
\begin{align}
    F(f_1,f_2)&=\frac{f_1f_2+\frac{1}{9}(1-f_1)(1-f_2)}{f_1f_2+\frac{1}{3}(f_1+f_2-2f_1f_2)+\frac{5}{9}(1-f_1)(1-f_2)} 
    \notag \\ 
    &=\frac{10f_1f_2-f_1-f_2+1}{8f_1f_2-2f_1-2f_2+5} 
    \label{eq:purif} \\
    P(f_1,f_2)&=f_1f_2+\frac{1}{3}(f_1+f_2-2f_1f_2)+\frac{5}{9}(1-f_1)(1-f_2) 
    \notag  \\
    &= \frac{8}{9}f_1f_2-\frac{2}{9}(f_1+f_2)+\frac{5}{9}
    \label{eq:purif-p}
\end{align}

To fully exploit the benefit of entanglement purification under limited quantum resources, two critical research questions should be carefully addressed.

\textbf{Q1}: How to schedule entanglement purification? We have demonstrated in Introduction that different entanglement purification scheduling policies may produce purified EPR pairs of different fidelity. Moreover, entanglement purification and swapping should be carefully orchestrated so as to optimize the overall performance. 

    \textbf{Q2}: How to integrate entanglement purification into entanglement provisioning and routing so as to build min-cost end-to-end entanglements with given QoS requirements in terms of fidelity and end-to-end entanglement throughput? We address this problem in the next section.

\subsection{Single-hop Entanglement Purification Scheduling}

We start by solving the one-hop entanglement purification scheduling problem defined below.

\begin{problem}
    Alice and Bob is connected by a quantum link $e$, over which they have established $N$ elementary EPR pairs of fidelity $f_e$. We seek an entanglement purification scheduling policy producing the maximal number of EPR pairs whose fidelity is at least $f_{\theta}$.
    \label{pb:purif-single-hop}
\end{problem}

To get more insight into the purification scheduling problem, we can model a purification schedule as a tree, which we term as purification tree, as exemplified in Figure~\ref{fig:tree}. We denote an entanglement mapping to an internal node resulting from a sub-tree with $k-1$ leaves as a level-$k$ entanglement. An elementary entanglement is a level-$1$ entanglement which is also a leaf. Therefore, the problem of purification scheduling maps to constructing the purification tree leading to maximal number of purified entanglements, subject to the fidelity constraint $f_{\theta}$. Clearly, enumerating all the possible trees is inefficient because there are an exponential number of possible trees, given that the number of non-isomorphic trees for a given number of nodes is the well-known Cayley number which scales exponentially w.r.t. the tree size. Therefore, we need a dedicated design, which is the central task of our work.

\begin{figure}[!ht]
\centering
\includegraphics[width=0.28\textwidth]{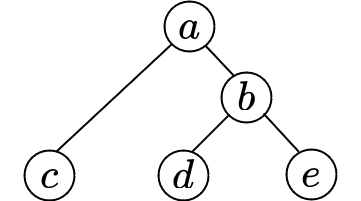}
\caption{Example of a purification tree}
\label{fig:tree}
\end{figure}

\textbf{Data structure}. We create a list ${\cal L}$ to store the candidate purification trees. Each entry $l$ of ${\cal L}$ is a quadruple representing a purification tree or sub-tree. Specifically, in the quadruple $l\triangleq (b_l,\hat{f_l},\hat{\xi}_l,T_l)$:
\begin{itemize}
    \item $b_l$ denotes the number of leaves in the tree;
    \item  $\hat{f_l}$ denotes the discretized fidelity of the root of the tree;
    \item $\xi_l$ denotes the expected percentage of the EPR pairs of fidelity $f_l$ generated in the purification scheduling process. In other words, there are in average $\xi_l N$ EPR pairs of fidelity $f_l$ produced in the purification scheduling process. $\hat{\xi}_l$ denotes the discretized value approximating  $\xi_l$, as detailed in our algorithm.
    \item $T_l$ records the entire purification tree or sub-tree $l$.
\end{itemize} 
Given two entries $l_i=(b_i,\hat{f_i},\hat{\xi}_i,T_i)$, $i=1,2$, we say that $l_1$ \textit{dominates} $l_2$ if $b_1\le b_2$, $\hat{f_1}\ge \hat{f_2}$ and $\hat{\xi}_1\ge \hat{\xi}_2$. The
dominance is strict if at least one inequality holds strictly. Intuitively, an optimal purification tree is not strictly dominated by any other tree. So we can safely remove all the dominated entries in $\cal L$. 

\textbf{Algorithm}. Our algorithm is depicted in Algorithm~\ref{alg:purif2} consisting of two major parts.

We first compute the lower-bound on the number of leaves of a purification tree such that any purification tree reaching this bound generates EPR pairs whose fidelity is at least $f_{\theta}$. To that end, we define $\Gamma(i)$, $1\le i\le N$, as the minimal fidelity achievable from $i$ EPR pairs of fidelity $f_e$. The boundary condition is given as $\Gamma(1)=f_e$. We can compute $\Gamma(i)$ by dynamic programming as below.
$$\Gamma(i)=\min_{1\le k\le \lfloor i/2\rfloor} F(\Gamma(k),\Gamma(i-k)).$$
It then follows that $N'=\argmin_i \Gamma(i)\ge f_{\theta}$ is the lower-bound we are looking for.

The core part of our algorithm is the second \textbf{for} loop computing $\cal L$. At a high level, by executing the loop, our algorithm builds and searches the purification trees from the leaves. 
As we will never purify any EPR pair whose fidelity already reaches $f_{\theta}$, the maximal number of leaves of the purification tree we search is upper-bounded by $2(N'-1)$, with the worst case being the situation where two EPR pairs, each representing a purification tree of $N'-1$ leaves are purified to produce an EPR pair of sufficient fidelity. In each iteration of the second for loop, we scan each pair of entries $l_1,l_2\in {\cal L}$ and build a new tree $l_3$ with $l_1$ and $l_2$ being its children\footnote{By slightly abusing terminology, we use the entry also to denote the purification tree corresponding to it.}. 
If $l_3$ is not dominated by any entry in ${\cal L}$, we add $l_3$ to ${\cal L}$. Noticing the discretization process to compute $\hat{\xi}_3$ and $\hat{f}$, for each pair of {$(b,\hat{\xi},\hat{f})$}, 
there are at most $\frac{1}{\Delta\xi\cdot\Delta f}$ entry in $\cal L$ as we remove all the dominated entries. Therefore, there are at most {$\frac{n}{\Delta\xi\cdot\Delta f}$} entries in $\cal L$, where $n=\min\{N,2(N'-1)\}$. Each entry containing a purification tree, the overall space complexity of our algorithm sums up to {$O\left(\frac{n}{\Delta\xi\cdot\Delta f}\right)$}. On the other hand, in each iteration, we grow the tree by merging two sub-trees, each mapping to an entry in $\cal L$. The overall time complexity sums up to {$O\left(\frac{n^3}{\Delta\xi^2 \Delta f^2}\right)$}.

It remains to analyze the optimality of Algorithm~\ref{alg:purif2}. The main source of efficiency loss comes from discretization which may accumulate along the purification tree from its leaves up to root. Theorem~\ref{thm:opt-single-hop} establishes that, with sufficient discretization granularity, Algorithm~\ref{alg:purif2} outputs a solution approaching optimum. \textcolor{black}{The detailed proof is shown in the appendix.}

\begin{algorithm}
\caption{Entanglement purification scheduling}
\label{alg:purif2}
\LinesNumbered 
\KwIn{$N$, $f_e$, $f_{\theta}$}
\KwOut{a purification tree}

\textbf{Initialization}: $\Gamma(i)\leftarrow f_e$, $1\le i\le N$, ${\cal L}\leftarrow\{(1,f_e,1,{T_1})\}$ 
        
    \For{$i=2$ \textbf{to} $N$, $k=1$ \textbf{to} $\lfloor i/2\rfloor$}
    {
        \If{$\Gamma(i)< F(\Gamma(k),\Gamma(i-k))$}
        {
            $\Gamma(i)\leftarrow F(\Gamma(k),\Gamma(i-k))$
        }
    }

    $N'\leftarrow \argmin_i \Gamma(i)\ge f_{\theta}$
 
    \For{$i=1$ \textbf{to} $\min\{N,2(N'-1)\}$}
    {
        
            \ForEach{pair of entries $l_i=(b_i,\hat{f_i},\hat{\xi}_i,T_i)\in {\cal L}$, $i=\{1,2\}$, satisfying $b_1+b_2\le \min\{N,2(N'-1)\}$} 
            {
                {$\displaystyle\hat{f}\leftarrow\left \lceil\frac{F(f_1,f_2)}{\Delta f}\right\rceil\Delta f$}

                {$\displaystyle\hat{\xi}_3\leftarrow \left\lceil\frac{P(f_1,f_2)\min\{\hat{\xi}_1,\hat{\xi}_2\}}{\Delta\xi}\right\rceil\Delta\xi$}
                

                create a new tree $T_3$ with the two children being $T_1$ and $T_2$ 
                
                $l_3\leftarrow (b_1+b_2,\hat{f},\hat{\xi}_3,T_3)$
                    
                \If{$l_3$ is not dominated by any entry in ${\cal L}$}
                {
                    add $l_3$ to ${\cal L}$
                }
            }    
        
    }
    \eIf{there exists at least an entry in $\cal L$ whose fidelity is at least $f_{\theta}$}
    {   
        let $l'\triangleq(b',\hat{f}',\hat{\xi}',T')$ denote the entry with maximal value of $\hat{\xi}'/b'$ among them
        
        \textbf{return} the purification tree $T'$
    }
    {
        \textbf{return} no feasible purification schedule
    }   
\end{algorithm}

\begin{theorem}
    Let $T'$ denote the tree output by our algorithm and {$l'=(b',\hat{f}',\hat{\xi}',T')$ denote the  corresponding entry in $\cal L$. Under the condition
    $\Delta\xi\le b'\xi'\epsilon$ and $\Delta f\le b'f'\epsilon$, Algorithm~\ref{alg:purif2} outputs an $\epsilon-$optimal solution of Problem~\ref{pb:purif-single-hop}.}
    \label{thm:opt-single-hop}
\end{theorem}

It follows from Theorem~\ref{thm:opt-single-hop} that $\Delta\xi$ depends on $\xi'$ and $b'$, which poses a problem, as we need to set $\Delta\xi$ before knowing $\xi'$ and $b'$. A practical way to set $\Delta\xi$ is to pick a purification schedule, e.g., \textsc{Parallel}, denoted by $l_0\triangleq (b_0,f_0,\xi_0,T_0)$ and set $\Delta\xi=b_0\xi_0\epsilon$. As our algorithm achieves quasi-optimality, it holds that $\xi'b'\simeq \hat{\xi}'b'\ge \xi_0 b_0$. Hence we have 
$\Delta\xi=b_0\xi_0\epsilon<b'\xi'\epsilon$. $\Delta f$ can be set similarly to satisfy the condition in Theorem~\ref{thm:opt-single-hop}.
\begin{remark}
    It is insightful to compare our results with the existing works~\cite{chenjiajsac} on the simple bit flip error model, where $$F(f_1,f_2)=\frac{f_1f_2}{f_1f_2+(1-f_1)(1-f_2)}.$$ 
In this case, by algebraically rewriting the above equation as
$$1-\frac{1}{F(f_1,f_2)}=\left(1-\frac{1}{f_1}\right)\cdot\left(1-\frac{1}{f_2}\right),$$
we can check that the entanglement purification map is associative. Hence, any entanglement purification schedule will output a purified EPR pair of the same fidelity. In contrast, we consider the generic Pauli channels, where the need for twirling makes the map~\eqref{eq:purif} non-associative~\cite{PhysRevA.59.169}, calling for optimal design of entanglement purification scheduling policy.
\label{remark:1}
\end{remark}

\subsection{Multi-hop Entanglement Purification Scheduling}

We tackle the more challenging multi-hop situation with Alice and Bob separated by quantum repeaters and rely on entanglement swapping to establish end-to-end EPR pair. Entanglement purification scheduling should be jointly optimized with entanglement swapping scheduling, formulated below.

\begin{problem}
    Given that Alice and Bob are connected by an entanglement path $P$, we seek an optimal joint entanglement purification and swapping scheduling policy maximizing the fidelity of the final EPR pair between them.
    \label{pb:swap}
\end{problem}

To solve Problem~\ref{pb:swap}, we investigate a scheduling policy called \textsc{purify-and-swap}, which consists of firstly running Algorithm~\ref{alg:purif2} to produce a purified elementary EPR pair of maximal fidelity at each hop and then performing entanglement swapping at each quantum repeater to stitch the purified elementary EPR pairs into an end-to-end one.

Structurally, \textsc{purify-and-swap} can be decomposed to two phases consisting of exclusively entanglement purification and entanglement swapping. Therefore, it is practically easy to deploy and tractable to analyze. In contrast, there are more complex policies where entanglement purification and swapping are intertwined in a nested way, rendering the deployment relatively complex and the performance analysis highly intractable. However, despite the neatness of \textsc{purify-and-swap}, its optimality cannot be guaranteed in general. Motivated by this observation, we investigate the following question: under what conditions are \textsc{purify-and-swap} optimal? Our main results in this subsection is the closed-form conditions to ensure its optimality.

Denote $P$ the entanglement path between Alice and Bob. Suppose we have established an EPR pair of fidelity $f_e$ over each link $e\in P$. We can compute the fidelity of the end-to-end EPR pair formed via entanglement swapping along $P$ as~\cite{chpt-purif}

\begin{eqnarray}
    f(P)=\frac{1}{4}\left(1+3\prod_{e\in P} \frac{4f_e-1}{3} \right).
    \label{eq:swap}
\end{eqnarray}
Mathematically, it is more convenient to write~\eqref{eq:swap} as
$$\frac{4f(P)-1}{3}=\prod_{e\in P} \frac{4f_e-1}{3}.$$

We first present an auxiliary lemma, whose proof, \textcolor{black}{consists of algebraic operations and is detailed in the appendix} 

\begin{lemma}
    Consider $3$ quantum nodes as shown in Figure~\ref{fig:puri-swap}, where there are two EPR pairs between Alice and Charlie of fidelity  $f_1$ and $f_2$, and two EPR pairs between and Bob and Charlie of fidelity $f_3$ and $f_4$. 
    \begin{itemize}
        \item If $f_1,f_2\ge 0.5$ and $f_3,f_4\ge 0.7$, then \textsc{purify-and-swap} outperforms \textsc{swap-and-purify} in terms of fidelity.
        \item If $f_i\ge 0.7$ for $1\le i\le 4$ and the swapping success probability is upper-bounded by $0.818$, then \textsc{purify-and-swap} outperforms \textsc{swap-and-purify} in success probability, i.e., the probability of establishing an end-to-end EPR pair between Alice and Bob.
    \end{itemize}
    \label{fact:aux}
\end{lemma}

\begin{remark}
    We make the following clarification:
    \begin{itemize}
        \item Regarding the fidelity, the condition is relatively tight such that $f_1=0.5$, $f_2=1$, $f_3=0.699$, $f_4=1$ forms a counter-example that \textsc{swap-and-purify} outperforms \textsc{purify-and-swap} by $0.00004$ in fidelity. When the condition does not hold, we report numerically by scanning $f_i\in[0.5,0.7]$ with stepsize $0.001$ that \textsc{purify-and-swap} outperforms \textsc{swap-and-purify} in $100$\% cases.
        \item Regarding swapping success probability, due to the limitation of current BSM scheme with linear optics, this probability does not exceed $0.5$, i.e., far less than $0.818$~\cite{Bayerbach2022BellstateME}, validating the condition empirically. As in the first case, we trace the success probability with $f_i\in[0.5,0.7]$ with stepsize $0.001$ and report that \textsc{purify-and-swap} always outperforms \textsc{swap-and-purify}.
    \end{itemize}
\end{remark}
 
Armed with Lemma~\ref{fact:aux}, we next state the main result of this subsection. The proof consists of showing that if \textsc{purify-and-swap} is outperformed by another policy, then we can construct a sub-instance of the system where \textsc{swap-and-purify} outperforms \textsc{purify-and-swap}, which leads to a contradiction with Lemma~\ref{fact:aux}. \textcolor{black}{The detailed proof is provided in the appendix}

\begin{theorem}
    If the fidelity of elementary EPR pairs is at least $0.7$ for each link of $P$ and the swapping success probability does not exceed 0.818, then \textsc{purify-and-swap} is optimal.
    \label{theorem:opt-purif-schedule}
\end{theorem}

\section{Fidelity-constrained Entanglement Routing}
\label{sec:rout}

We have addressed the first research question regarding entanglement scheduling along a predetermined entanglement path. We proceed to the second question on how to integrate entanglement purification into entanglement provisioning and routing so as to build min-cost end-to-end entanglement paths with given QoS requirements. We first investigate the case of single Alice-Bob pair and then tackle multi-flow path optimization.

\subsection{The Case of Single Alice-Bob Pairs}
We model the quantum network as an undirected graph $G\triangleq (V,E)$, where $V$ denotes the set of quantum nodes, i.e., quantum terminals and quantum repeaters, $E$ denotes the set of quantum channels. We note that a quantum channel between a pair of nodes $u$ and $v$ induces a communication channel if and only if $u$ and $v$ successfully establish a quantum entanglement. Each node $v$ disposes $Q_v$ qubits, allowing it to create at most $Q_v$ quantum entanglements with its neighbors. The aforementioned quantum network model is essentially the same as the one in research~\cite{10.1145/3341302.3342070,shiton}, but we added considerations for fidelity, purification and swapping.
Let $p_v$ denote the entanglement swap success probability of node $v$. As the qubits are limited resources, quantum entanglements need to be carefully provisioned to support communication sessions among quantum nodes. 

Consider a communication session between Alice and Bob requesting $q_0$ qubits of fidelity at least $f_0$. Let $\cal P$ denote the set of entanglement paths between them. Given an edge $e\in P\in {\cal P}$, we define a non-decreasing cost function $C_e(m_e)$, where $m_e$ denotes the number of elementary EPR pairs over $e$ that are used to build the end-to-end EPR pair along $P$. The cost of path $P$ is defined as $C(P)\triangleq\sum_{e\in P}C_e(m_e)$. Below we give two concrete examples. 
\begin{itemize}
    \item It is often meaningful to minimize the number of consumed EPR pairs to achieve a given end-to-end fidelity. In this case, we can simply set $C_e=m_e$ or, in the weighted case, $C_e=w_em_e$, with $w_e$ being the associated weight.
    \item When multiple EPR pairs are involved over an edge, entanglement purification needs to be executed in multiple rounds. In this case, $C_e(m_e)$ can model the delay to obtain a purified EPR pair over $e$. $C(P)$ is thus the entire delay to establish the end-to-end EPR pair.
\end{itemize}

We call an entanglement path \textit{feasible} if the following conditions are satisfied.
\begin{itemize}
    \item The corresponding end-to-end EPR pairs are constructed by \textsc{purify-and-swap}, where entanglement purification follows Algorithm~\ref{alg:purif2}, subject to the qubit resource limit of each quantum node;
    \item The expected number of successfully established end-to-end EPR pairs is at least $q_0$;
    \item The fidelity of the end-to-end EPR pairs is at least $f_0$.
\end{itemize}

We seek to solve the following QoS-constrained entanglement path optimization problem.

\begin{problem}
    The min-cost feasible entanglement path optimization problem seeks a feasible entanglement path between Alice and Bob of minimal cost.
    \label{pb:rout}
\end{problem}

\begin{remark}
    By seeking an optimal entanglement path we mean to find not only the path, but also the number of EPR pairs over each edge of the path involved in entanglement purification.  
\end{remark}

Problem~\ref{pb:rout} is \textsc{NP}-hard, as stated in Lemma~\ref{lemma:np}. To prove the hardness, we can focus on a degenerated problem instance where $C_e(m_e)=w_em_e$ and each quantum repeater disposes two qubits, rendering entanglement purification impossible, because it needs to establish two EPR pairs with its preceder and successor in path $P$, leaving no more qubits for entanglement purification. $q_0$ is set to a very small real number that that any path satisfies the entanglement throughput constraint. In this situation, both path cost and fidelity can be formulated as additive metrics. The problem thus degenerates to the constrained min-cost path problem which is \textsc{NP}-hard.

\begin{lemma}
    The fidelity-constrained min-cost entanglement path problem is \textsc{NP}-hard.
    \label{lemma:np}
\end{lemma}

Given the hardness result, we focus on approximation algorithm design. A technical challenge distinguishing Problem~\ref{pb:rout} from the constrained path optimization problems studied in the literature lies in the feasibility requirement, where the number of qubits used at each node $v$ cannot exceed its capacity $Q_v$. 

To model the feasibility constraint, we construct an auxiliary graph $G'\triangleq (V',E')$ based on $G$. 
\begin{itemize}
    \item \textbf{Vertices}: For each vertex $v\in V-\{s,t\}$, we create $Q_v-1$ vertices $v_i$, $i\in[1,Q_v-1]$. $v_i$ corresponds to $v$ disposing $i$ available qubits for entanglement. Note that any $v\in V-\{s,t\}$ is a quantum repeater connecting two nodes in the $st-$path, $v$ can take at most $Q_v-1$ qubits to entangle with any neighbor. We then create $Q_s$ vertex $s_i$, $i\in[1,Q_s]$, representing $s$ and $Q_t$ vertices $t_i$, $i\in[0,Q_t-1]$, representing $t$. Finally, we create a virtual source node $s'$ and virtual destination node $t'$.
    \item \textbf{Edges}: For each edge $(u,v)\in E$, we create edges $(u_i,v_j)$ if $Q_v-j \le i $ and $Q_v-j\le C_{l}$, where $C_{l}$ is link capacity.The rationale behind our construction is below: $(u_i,v_j)$ represents $Q_v-j$ EPR pairs between $u$ and $v$. This is feasible as (1) $u_i$ implicates that $u$ disposes $i$ qubits; (2) since $v_j$ represents $v$ disposing $j$ available qubits, it can establish $Q_v-j$ EPR pairs with $u$ as long as $Q_v-j\le i$ and $Q_v-j\le C_{l}$. We set the cost of $(u_i,v_j)$ as
    \vspace{-0.5cm}
    $$C'_{(u_i,v_j)}=C_{(u,v)}(Q_v-j).$$
    Finally, we create an edge between $s'$ and $s_i$ for each $1\le i\le Q_t$, $t_i$ and $t'$ for each $0\le i\le Q_t-1$ and set the edge cost to zero and fidelity to one.
\end{itemize}

\begin{figure}[!ht]
\centering
\includegraphics[width=0.48\textwidth]{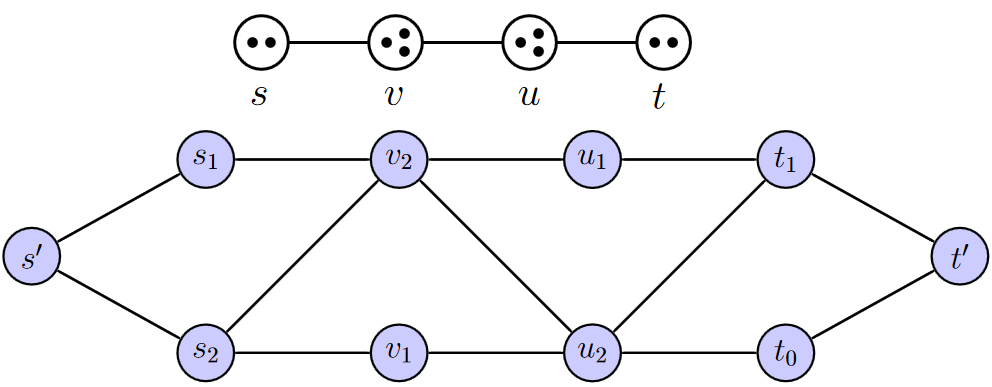}
\caption{An example of $G$ (upper) and $G'$ (lower): $s$ and $t$ dispose $2$ qubits, while $v$ and $u$ dispose $3$. Path $s'-s_2-v_1-u_2-t_0-t'$ corresponds to path $s-v-u-t$ by purifying $2$ EPR pairs at the edges $(s,v)$ and $(u,t)$.}
\label{fig:aux-graph}
\end{figure}

The following lemma follows from our construction of $G'$.

\begin{lemma}
    There is a bijective map between each feasible path from $s'$ to $t'$ and a feasible path between $s$ and $t$ in $G$.
\end{lemma}

Our fidelity-constraint path optimization algorithm, whose pseudo-code is given in Algorithm~\ref{alg:mincost}, consists of working on $G'$ and then transform the output path to a feasible path in $G$. To facilitate algorithm design, given a path $P$ of fidelity $f(P)$, we define its pseudo-fidelity $\phi(P)$ as 
$$\phi(P)\triangleq \ln\frac{4f(P)-1}{3}.$$
The introduction of pseudo-fidelity allows us to treat the path fidelity as an additive metric. We also transform the fidelity threshold $f_0$ to the corresponding pseudo-fidelity threshold $\phi_0$.

Moreover, let $\lambda_e$ denote the expected number of purified EPR pairs produced over edge $e\in P$, and let $\lambda(P)$ denote the expected end-to-end EPR pairs generated along $P$ by letting the quantum routers perform entanglement swapping. We have 
$$\lambda(P)=\min_{e\in P} \quad \lambda_e\!\!\!\!\!\prod_{v\in P, v\ne \{s,t\}}\!\!\!\!\! p_v.$$
We note that the above computed $\lambda(P)$ actually gives its lower-bound, as if we carefully optimize the order, under which entanglement swapping is performed at the intermediate routers, we can achieve better end-to-end throughput~\cite{order_matter}. The price to pay is the additional delay to establish the end-to-end entanglements, as we need to separate the swapping operations into rounds. Our algorithm can be applied to such optimization as long as the optimized throughput can be expressed as a function of $P$ and $\lambda_e$ along $P$.

To turn $\lambda(P)$ into an additive metric, we define the log-throughput $\psi_e\triangleq \ln \lambda_e$, {the log-swap-success-probability of node $v$ is termed as $\psi_v\triangleq \ln p_v$.} Accordingly, we have 
$$\psi(P)=\min_{e\in P} \psi_e+\sum_{v\in P, v\ne \{s,t\}} {\psi_v}$$

\textbf{Data structure}. We maintain a list ${\cal L}_v$ at each vertex $v\in V'$ containing quintuples as entries. Each quintuple $l$ represents a feasible path $P_l$ from $s'$ to $v$ that is not dominated by any other quadruple. Specifically, in a quintuple $l\triangleq (c,\hat{\phi},\psi_b,\hat{\psi},path)$, $c$ denotes the cost of $P_l$, $\hat{\phi}$ denotes the discretized pseudo-fidelity of $P_l$ with stepsize $\Delta\phi$, $\psi_b=\min_{e\in P}\{\psi_e\}$ denotes the bottleneck edge log-throughput of $P_l$, $\hat{\psi}$ denotes the discretized log-throughput of $P$ with stepsize $\Delta\psi$. The algorithm operates in $G'$, and we need to translate the sequence of visited nodes of path in $G'$ back to $G$, recording this sequence in the path list for loop-free. The entries are indexed by the couple $(\hat{\phi},\hat{\psi})$. Given two entries $l_i=(c_i,\hat{\phi}_i,\psi_{b,i},\hat{\psi}_{l,i},path)$, $i=1,2$, we say that $l_1$ dominates $l_2$ if $c_1\le c_2$, $\hat{\phi}_1\ge \hat{\phi}_2$, and $\hat{\psi}_{l,1}\ge \hat{\psi}_{l,2}$. The dominance is strict if at least one inequality holds strictly. Intuitively, an optimal feasible path is not strictly dominated by any feasible path. So we can safely remove dominated triples. 

\textbf{Algorithm}.
The core part of Algorithm~\ref{alg:mincost} is the main loop that iteratively scans all the vertices and their non-dominated labels using a min-heap priority queue $\Gamma$ and updates the corresponding entries. Initially, the queue contains the labels of $s$. The algorithm iteratively processes labels from the queue. For each label $l$ dequeued from $\Gamma$, the algorithm updates the visited label set associated with node $v$, denoted by $\Omega_v$. once the current node $u_i$ is the target node $t'$, meaning that the algorithm find the minimum cost from $s$ to $t$. The algorithm terminates once $u_i=t'$ or $\Gamma$ becomes empty. Update labels procedure involves expanding the label of current node to neighboring nodes. For each neighbor $v_j$ of the current node $u_i$, new labels are generated by iterating over all possible possible discrete values of the fidelity threshold $\hat{\phi}'$, i.e., $\frac{\hat{\phi}}{\Delta\phi}< k \le\left\lfloor\frac{\hat{\phi}-\phi_0}{\Delta\phi}\right\rfloor$. Then for each $\hat{\phi}'$, we perform the following actions.

\begin{itemize}
    \item We invoke Algorithm~\ref{alg:purif2} with $Q_v-j$ EPR pairs and fidelity threshold $\hat{\phi}'$ to compute the throughput $\psi_e$, which is the maximal throughput over $e$ satisfying the fidelity constraint, and this procedure can be pre-computed .   
    \item We then update the discretized throughput of the path $P$ concatenated by $e$ by comparing the bottleneck throughput of $P$, $\psi_b$, with $\psi_e$ and updating $\hat{\psi}'$ accordingly.
    \item If the log-throughput of the new path satisfies the throughput constraint, we then create a quintuple $l'$ corresponding to the new path and add it to ${\cal L}_{v_j}$ if it is not dominated by any other entry.
    \item By the above steps, we maintain in ${\cal L}_{v_j}$ the min-cost feasible $sv_j-$path for each $(\hat{\phi},\hat{\psi})$ couple.  
    \item Once the scan and update process is terminated, Algorithm~\ref{alg:mincost} selects the min-cost entry in ${\cal L}_{t'}$ and outputs the corresponding path $P$ and the related entanglement purification scheme. As the vertices $t_i$, $1\le i\le Q_t$, map to $t$, the output is an $st-$path.
\end{itemize}

\color{black}
We next prove that Algorithm~\ref{alg:mincost} outputs an $st-$path that can be arbitrarily close to optimum.
\begin{definition}
    Let $P^*$ denote the optimal path solving Problem~\ref{pb:rout}. An $st-$path $P$ is called an $\epsilon$-optimal feasible path if the following inequality holds.
    $$C(\hat{P})\le C(P^*), \ \phi(\hat{P})\ge(1-\epsilon)\phi_0, \ \psi(\hat{P})\ge (1-\epsilon)\psi_0.$$
\label{def:appro}
\end{definition}
The rationale behind the above definition is to allow a small quantity $\epsilon$ of overflow on the fidelity and throughput constraints. Theorem~\ref{theorem:opt-path}, whose proof is omitted here, proves that Algorithm~\ref{alg:mincost} outputs an $\epsilon$-optimal feasible path. 
\begin{theorem}
    If {$\Delta \phi\le \epsilon\phi_0/|V|$} and $\Delta\psi\le \epsilon \psi_0/|V|$, Algorithm~\ref{alg:mincost} outputs an $\epsilon$-optimal feasible path.
    \label{theorem:opt-path}
\end{theorem}

\begin{algorithm}
    \KwIn{$G'$, $s'$, $t'$, $\psi_0=\ln q_0$, $\phi_0=\phi(f_0)$ }
    \KwOut{a feasible entanglement path $P$}
    \ForEach{$v\in V'$} {
        \eIf{$v=s'$} {
            $l\leftarrow(0,0,\infty,\infty,\{s\})$ 

            ${\cal L}_v\leftarrow\{l\}$  

            $\Gamma.enqueue((c,l,v))$
        }
        {${\cal L}_v\leftarrow\emptyset$}
        
        $\Omega_v \leftarrow \emptyset$ 
    }    

    \While{$\Gamma$ is not empty}{
        $(c,l,u_i)\leftarrow \Gamma.dequeue()$

        \If{$l$ not in $\Omega_{u_i}$}{
        add $l$ to $\Omega_{u_i}$
        
        \If{$u_i=t'$}{
        break
        }
        {
        
        \ForEach{$v_j$ in $neighbor(u_i)$}{

            \If{$v$ not in path}{            
                \For{$\displaystyle k=1$ \textbf{to} $\displaystyle {\left\lfloor\frac{\hat{\phi}-\phi_0}{\Delta\phi}\right\rfloor}$}{
                    {$\hat{\phi}'\leftarrow \hat{\phi}-k\Delta\phi$}  
                    
                    {invoke Algorithm~\ref{alg:purif2} with $Q_v-j$ EPR pairs and threshold $-k\Delta\phi$ to compute the throughput $\psi_e$, so $\psi_e=\ln ((\hat{\xi}'/b')\cdot(Q_v-j)) $}
                

                    \eIf{$\psi_e\le \psi_b$}
                    {
                        $\displaystyle\hat{\psi}'\leftarrow \left\lceil \frac{{\psi_{v_j}}+\hat{\psi}+\psi_e-\psi_b}{\Delta\psi}\right\rceil\cdot \Delta\psi$; 
                    }
                    {
                        $\displaystyle\hat{\psi}'\leftarrow \left\lceil\frac{{\psi_{v_j}}+\hat{\psi}}{\Delta\psi}\right\rceil\cdot\Delta\psi$;
                    }
                
                    \If{$\hat{\psi}'\ge \psi_0$}{
                        $l'\leftarrow (c+C_e(Q_j-j),{\hat{\phi}'},\min\{\psi_e,\psi_b\},\hat{\psi}',path+[v])$\;
                        \If{$l'$ is not dominated by any entry in ${\cal L}_{v_j}$}
                        {
                            add $l'$ to ${\cal L}_{v_j}$,
                            \label{line:2}
                            add $l'$ to $\Omega_{v_j}$\;
                            $\Gamma.enqueue((c+C_e(Q_j-j), v_j, l'))$\;

                       }
                    }
            }
            }
        }
        
        }

    }

    \eIf{no entry in ${\cal L}_{t'}$ maps to a feasible path}{
        \Return no feasible path\;
    }{
        select the entry in ${\cal L}_{t'}$ with minimum cost, the selected entry has $st$-path $P$\; 
        \label{line:k}
        \Return $P$\;
    }
    \caption{computing fidelity-constrained min-cost path}
    \label{alg:mincost}
}
\end{algorithm}

\textbf{Space and time complexity of Algorithm~\ref{alg:mincost}}. Recall the construction of $G'$, each node $v\in G$ generates $Q_v$ nodes in $G'$. For each node in $G'$, we need to maintain a list of at most $\frac{\phi_0\psi_0}{\Delta\phi\Delta\psi}$ entries. Invocation of Algorithm~\ref{alg:purif2} can be pre-computed and stored for each pair of number of qubits and threshold, taking $O(\phi_0\max_v Q_v/\Delta\phi)$ space.
Noticing $|V'|=O(Q)$ and $|E'|=O(Q^2)$ with $Q\triangleq\sum_{v\in V}Q_v$ and setting $\Delta\phi$ and $\Delta\psi$ based on Theorem~\ref{theorem:opt-path}, the total space overhead sums up to $O(|V|^2Q\epsilon^{-2})$. {The time complexity is dominated by the main loop, with the heap manipulations, which sums up to $O\left(\left(|V'|+|E'|\right)\log V'\frac{\phi_0\psi_0}{\Delta\phi\Delta\psi}\right)$,} i.e., $O\left(Q^2\log Q |V|^2\epsilon^{-2}\right)$. Though polynomial, the time complexity can be significantly high for large $Q$. In this setting, we can decrease the size of the auxiliary graph to reduce the complexity in the following way. For each node, instead of creating $Q_v-1$ vertices in the auxiliary graph, we set a stepsize $\Delta Q$ and build $Q_v/\Delta Q$ nodes in $G'$, thus reducing the graph size by factor $\Delta Q$, with $\Delta Q=1$ being the degenerated case we analyzed in this subsection. The price to pay is the loss of efficiency in the solution we find, as we increase the granularity from $1$ to $\Delta Q$. In practice, we can tune $\Delta Q$ to trade off the efficiency against complexity. 

\subsection{The Case of Multiple Alice-Bob Pairs}
In this section, we investigate the problem of multi-flow path optimization, formulated below.
\begin{problem}
    Given a set $\cal K$ of $K$ Alice-Bob pairs with fidelity requirement $f_{0,k}$ and weight $w_k$ for each flow $k\in {\cal K}$, we seek a set of fidelity-constrained min-cost paths for a subset of flows maximizing their total weight.
    \label{pb:multi-rout}
\end{problem}

We develop a two-step solution to Problem~\ref{pb:multi-rout}. Step 1 extends Algorithm~\ref{alg:mincost} to compute a pool of candidate entanglement paths for each flow. Step 2 further selects at most one path from the pool of each flow to maximize the total weight. Technically, we develop an algorithmic framework to solve the corresponding combinatorial optimization problem.

In Step 1, we extend Algorithm~\ref{alg:mincost} to compute $R_k$ $\epsilon-$fidelitous feasible paths of minimal cost for each flow $k\in{\cal K}$. To this end, we modify Line~\ref{line:2} such that, if $l'$ is dominated by less than $R_k$ entries in ${\cal L}_{v_j}$, we add $l'$ to ${\cal L}_{v_j}$. We then modify Line~\ref{line:k} such that we select $R_k$ entries of minimum cost and trace back to obtain $R_k$ paths. If there are less than $R_k$ $\epsilon-$fidelitous feasible paths, we return all of them.

In Step 2, we formulate Problem~\ref{pb:multi-rout} as the following integer linear programming (ILP).
{\small
\begin{align*}
    \text{maximize} \quad & \sum_{k\in{\cal K}} \sum_{i\in{\cal R}_k} w_kx_{ki} &  \\
    \text{subject to} \quad 
                    & \sum_{k\in{\cal K}}\sum_{v\in P_{ki}} a_{kiv}x_{ki} \le Q_{v} &  \forall v\in V \\
                    & \sum_{k\in{\cal K}}\sum_{l\in P_{ki}} b_{kil}x_{ki} \le C_{l} &  \forall l\in E \\
                    & \sum_{i\in {\cal R}_k} x_{ki}\le 1 &  \forall k\in {\cal K} \\
                    & x_{ki}\in \{0,1\} &  \forall k\in {\cal K}, i\in {\cal R}_k.
\end{align*}}

In the formulated ILP, the binary variable $x_{ki}$ indicates whether path $i$ in the candidate path pool of flow $k$, denoted by ${\cal R}_k$, is selected, $a_{kiv}$ denotes the total number of qubits consumed at node $v$ by path $i$ in ${\cal R}_k$, denoted by $P_{ki}$. Note that $a_{kiv}$, $b_{kil}$ can be computed from the output of the extended version of Algorithm~\ref{alg:mincost} in Step 1. The first constraint states that the total number of qubits consumed cannot exceed the available qubit resources at each node. The second constraint states that the total number of EPR pairs consumed cannot exceed the link capacity at each edge. The third constraint states that at most one path is selected for each flow.
\color{black}

It is worth noting that our ILP belongs to a generic class of combinatorial optimization problem called \textit{packing integer program}, to which no constant-factor approximation solution is reported in the literature~\cite{doi:10.1137/S0097539799356265}. We emphasize that our ILP is not a Knapsack problem as the number of constraints is not constant, but scales as the problem size. It is exactly due to the large number of constraints that makes our ILP challenging to solve, compared to the Knapsack problem where there exist a palette of pseudo-polynomial-time algorithms.

Driven by the above observation, we formulate an LP by relaxing the above ILP and introducing a discount coefficient $\beta\in(0,1)$ on the first constraint, as detailed below.
{\small
\begin{align*}
    \text{maximize} \quad & \sum_{k\in{\cal K}} \sum_{i\in{\cal R}_k} w_kx_{ki} &  \\
    \text{subject to} \quad 
                    & \sum_{k\in{\cal K}}\sum_{v\in P_{ki}} a_{kiv}x_{ki} \le \beta Q_{v} &  \forall v\in V \\
                    & \sum_{k\in{\cal K}}\sum_{l\in P_{ki}} b_{kil}x_{ki} \le C_{l} &  \forall l\in E \\
                    & \sum_{i\in {\cal R}_k} x_{ki}\le 1 &  \forall k\in {\cal K} \\
                    & x_{ki}\ge 0 &  \forall k\in {\cal K}, i\in {\cal R}_k
\end{align*}}
The constraint $x_{ki}\le 1$ is implicit by the second constraint.

Let $x_{ki}^*$ denote the solution of the relaxed LP. For each flow $k$ where $\sum_{i\in{\cal R}_k} x_{ki}^*>0$, we set $\hat{x}_{ki}=1$ with probability $x^*_{ki}$ such that there is at most one non-zero $\hat{x}_{ki}$ for all $i\in{\cal R}_i$. This can be done by arbitrarily sorting $x_{ki}^*$, then picking a random number $x\in[0,1]$, finally setting $\hat{x}_{ki}=1$ if $x$ falls in the interval {$\left[\sum_{j\in{\cal R}_k, j<i} x_{kj}^*, \sum_{j\in{\cal R}_k, j\leq i} x_{kj}^*\right)$}.
 

Our main efforts are then to prove that $\mathbf{\hat{x}}\triangleq\{\hat{x}_{ki}\}$ approaches the optimal solution of the original ILP.

We first prove an auxiliary lemma. The proof, presented in the appendix, consists of algebraic and stochastic demonstrations.

\begin{lemma}
    $\forall \delta>0$, it holds that 
    $$\Pr\left[\sum_{k\in{\cal K}}\sum_{v\in P_{ki}} \!\!\!a_{kiv}\hat{x}_{ki}>(1+\delta)\beta Q_v\right]<\left[\frac{e^\delta}{(1+\delta)^{(1+\delta)}}\right]^{\beta Q_v}.$$
    \label{lemma:aux3}
\end{lemma}

Armed with Lemma~\ref{lemma:aux3}, we are able to derive our main result.

\begin{theorem}
    $\forall \epsilon>0$, let $\beta=1-\epsilon$, it holds that $\mathbf{\hat{x}}$ is a $2\epsilon-$optimal solution of the original ILP with prob. $\ge\frac{1}{3}$ if 
    \begin{equation}
        Q_v\ge \frac{\ln(3|V|)}{(1-\epsilon)\epsilon^2} \ \forall v\in V \ \text{and} \ \epsilon^2(1-\epsilon)\min_{k\in{\cal K}}w_{k}\ge \ln3. 
        \label{eq:cond-mult-flow}
    \end{equation}
    \label{theorem:mult-flow}
\end{theorem}

Theorem~\ref{theorem:mult-flow} proves that the number of qubits at each node only needs to scale logarithmically in the network size so as for our randomized algorithm to approach to optimal performance with prob. $\ge 1/3$. The second constraint in \eqref{eq:cond-mult-flow} can be satisfied by scaling each $w_k$. To further obtain a $2\epsilon-$optimum of the original ILP with probability at least $1-\delta$ in the standard algorithmic sense, it suffices to run $\log_3 \delta^{-1}$, i.e., $O(\ln \delta^{-1})$, parallel instances of our randomized algorithm and taking the solution of maximal total weight.

\section{Performance Evaluation}
\label{sec:simu}
\textcolor{black}{We study the scheduling and routing algorithms at the network layer of quantum networks. Routing algorithms leverage the link layer model~\cite{10.1145/3341302.3342070,2019Routing}. As mentioned in research~\cite{shiton,1020231}, entanglement routing can be evaluated through simulated experiments as long as the physical and link layers accurately reflect the physical facts.} 

Therefore, we conduct two sets of simulations to empirically evaluate our entanglement  purification scheduling and routing algorithms so as to get more insights on these two fundamental problems that cannot be fully captured by our analytical demonstrations presented in the paper.

\subsection{Entanglement Purification Scheduling}
Our first set of simulations focus on entanglement purification scheduling investigated in Section~\ref{sec:purif}. For single-hop scheduling, we compare Algorithm~\ref{alg:purif2} with the two purification strategies depicted in Figure~\ref{fig:puri-sch}, with the left strategy following a symmetrical pattern by parallelizing the purification process and the right strategy adopting a serial and recursive pattern. We term them as \textsc{symmetric} and \textsc{pumping} strategies, respectively. Figure~\ref{fig:simu-single-hop} traces the fidelity of purified EPR pair with different number of input EPR pairs. We observe that for \textsc{symmetric} and \textsc{pumping}, the fidelity gain is much less pronounced when the number of input EPR pairs is large. The output fidelity becomes almost constant after 10 input EPR pairs. In contrast, our algorithm is still able to achieve a relatively large fidelity gain. To better understand this phenomenon, we examine eq.~\eqref{eq:purif} and find that the fidelity gain decreases with input fidelity. Therefore, simple recursive purification scheduling strategies can bring significant fidelity gain at the first few rounds but can hardly increase the fidelity afterwards. \textcolor{black}{Figure~\ref{fig:simu-single-hop-PROB} traces the purification success probability of purified EPR pair with different number of input EPR pairs. Similar to the purified fidelity, the purification success probability exhibits a comparable behavior.} Our simulation results thus demonstrate the importance of designing an optimal strategy as Algorithm~\ref{alg:purif2}.

\begin{figure*}[t]
\centering
\subfigure{
\begin{minipage}[t]{0.31\textwidth}
\includegraphics[width=\textwidth]{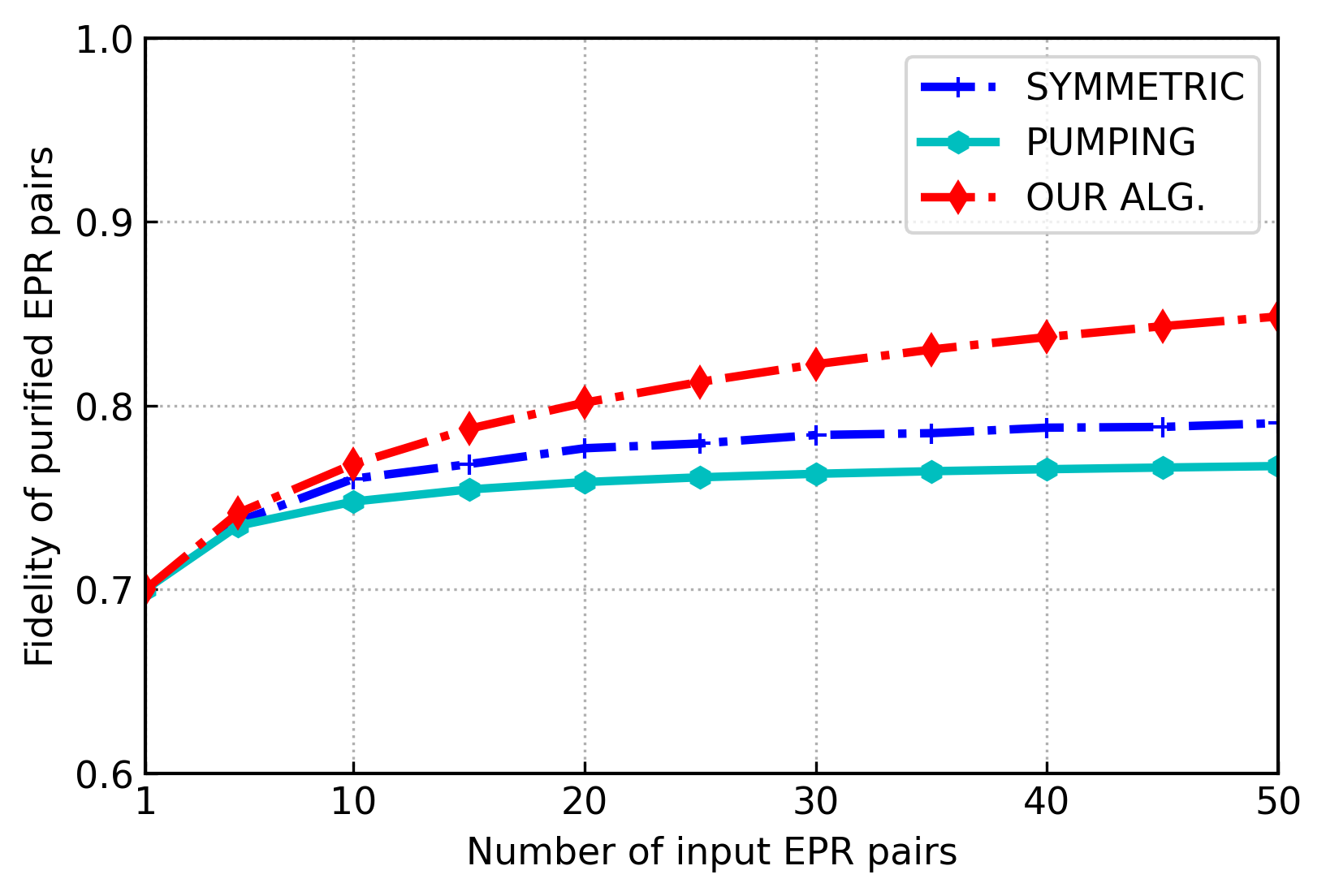}
\end{minipage}}\hspace{0.2cm}
\subfigure{
\begin{minipage}[t]{0.31\textwidth}
\includegraphics[width=\textwidth]{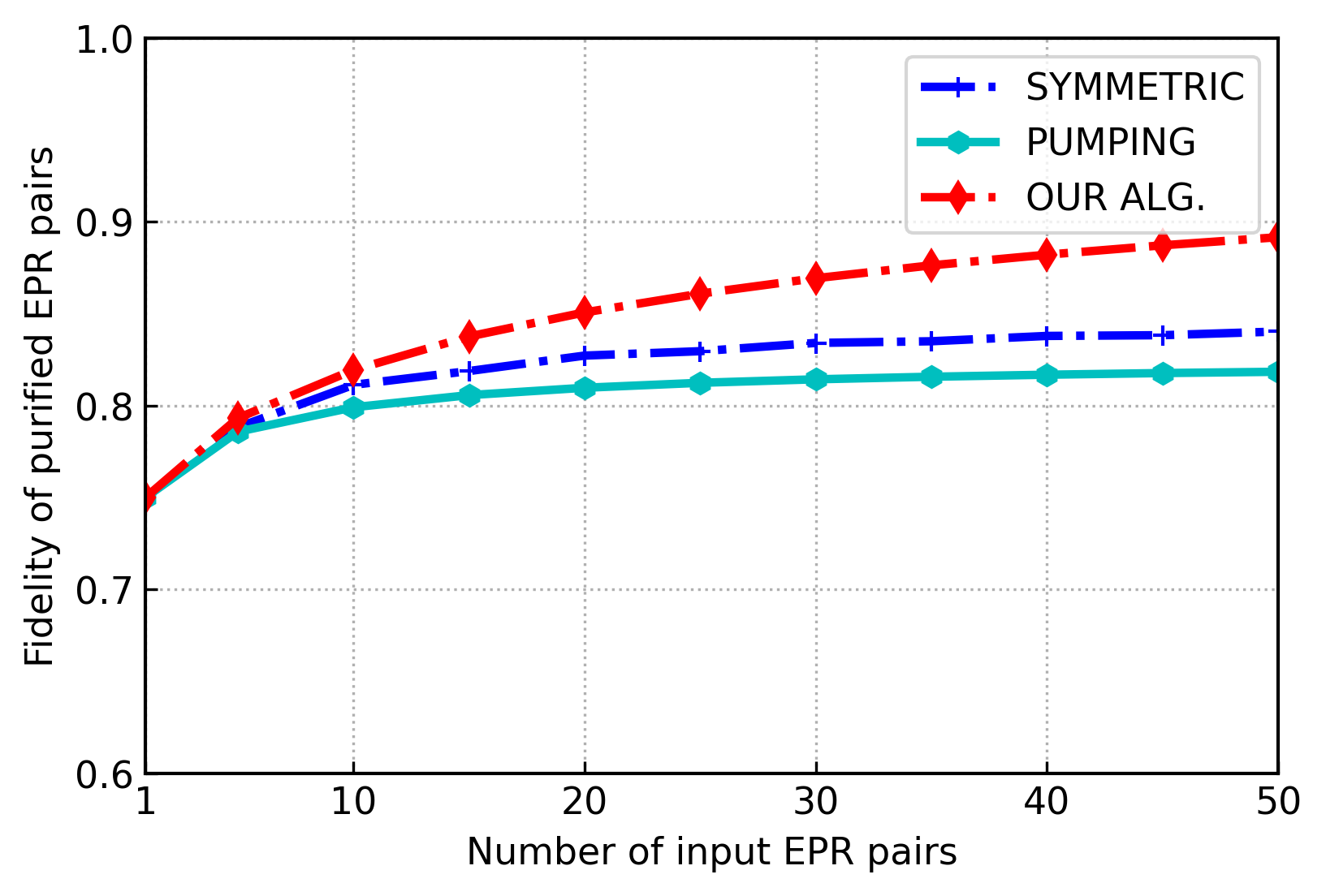}
\end{minipage}}\hspace{0.2cm}
\subfigure{
\begin{minipage}[t]{0.31\textwidth}
\includegraphics[width=\textwidth]{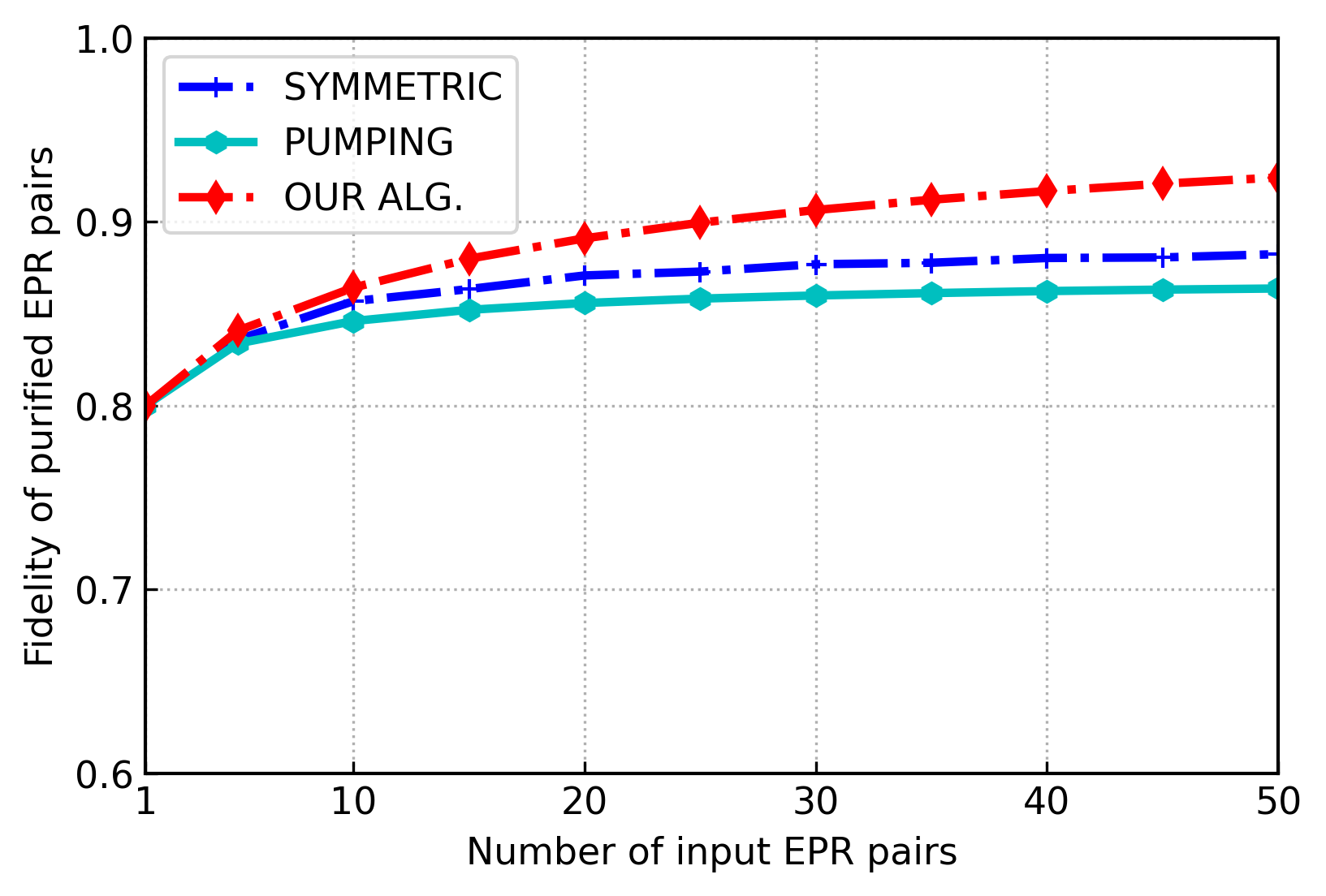}
\end{minipage}}
\vspace{-0.4cm}
\caption{Purified fidelity comparison of single-hop entanglement purification scheduling algorithms: the fidelity of elementary EPR pairs is 0.7, 0.75, and 0.8 in the left, middle, and right figures.} \label{fig:simu-single-hop}

\subfigure{
\begin{minipage}[t]{0.31\textwidth}
\includegraphics[width=\textwidth]{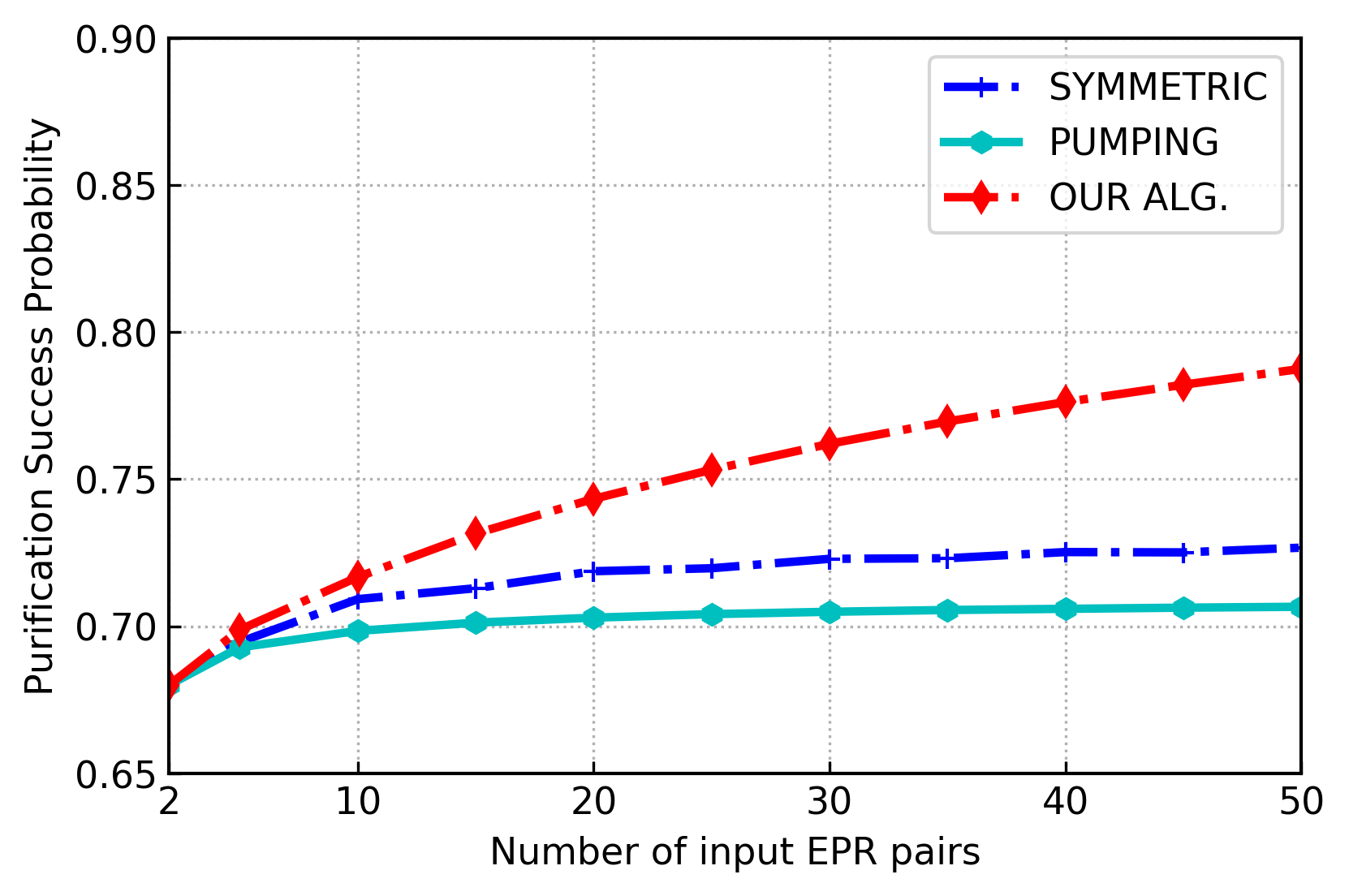}
\end{minipage}}\hspace{0.2cm}
\subfigure{
\begin{minipage}[t]{0.31\textwidth}
\includegraphics[width=\textwidth]{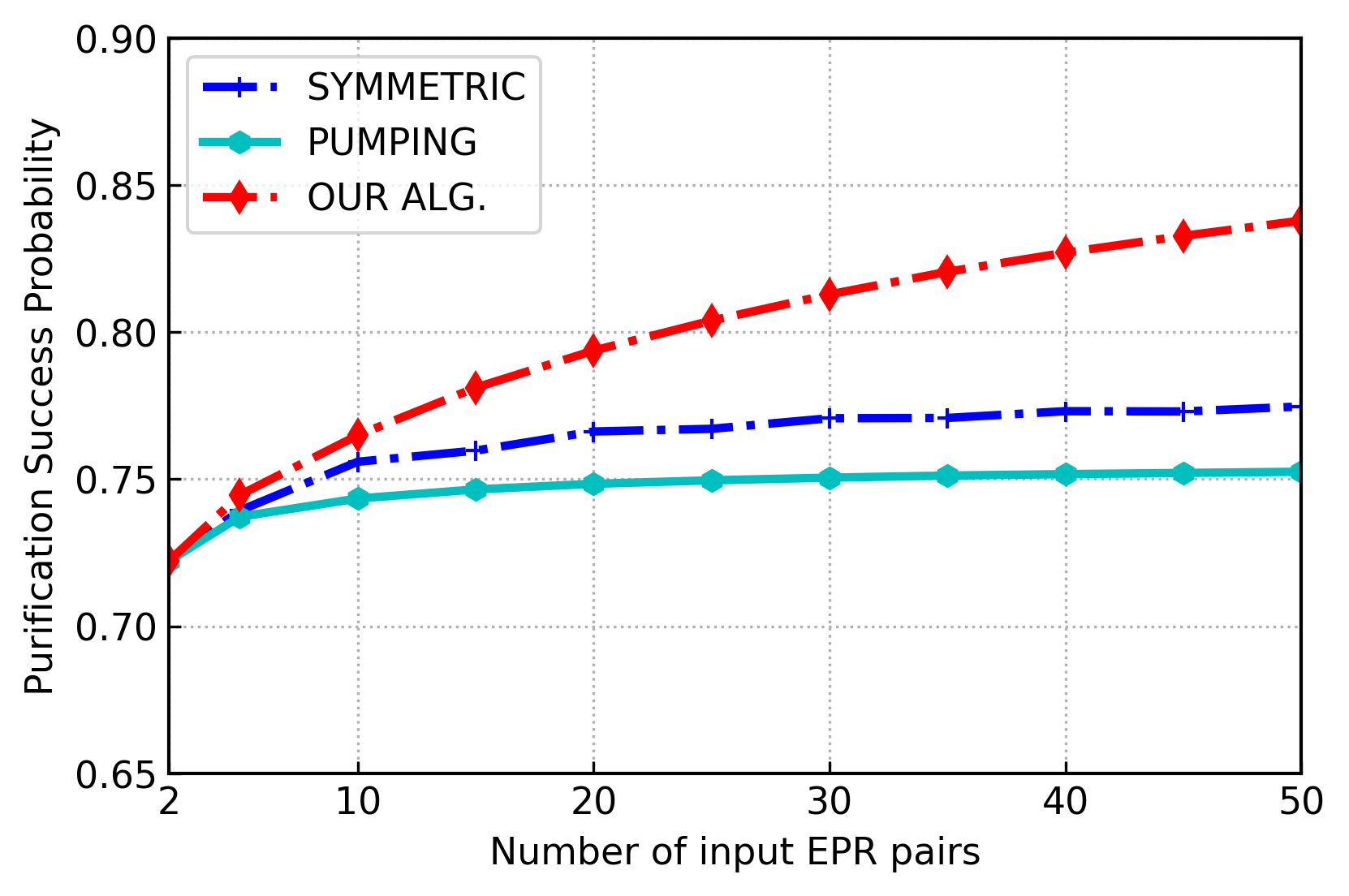}
\end{minipage}}\hspace{0.2cm}
\subfigure{
\begin{minipage}[t]{0.31\textwidth}
\includegraphics[width=\textwidth]{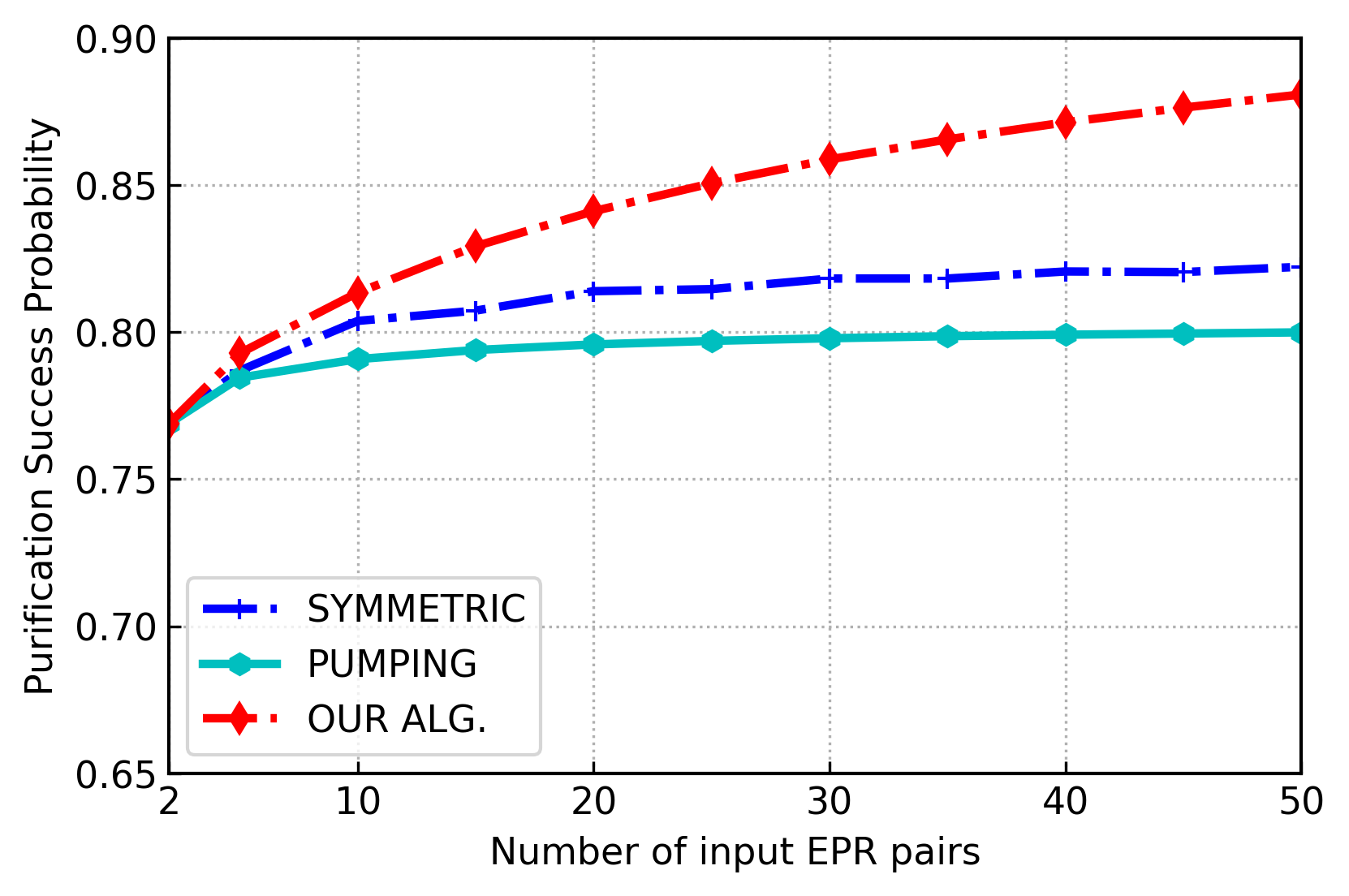}
\end{minipage}}
\vspace{-0.4cm}
\caption{Purification success probability comparison of single-hop entanglement purification scheduling algorithms: the fidelity of elementary EPR pairs is 0.7, 0.75, and 0.8 in the left, middle, and right figures.} \label{fig:simu-single-hop-PROB}

\subfigure{
\begin{minipage}[t]{0.31\textwidth}
\includegraphics[width=\textwidth]{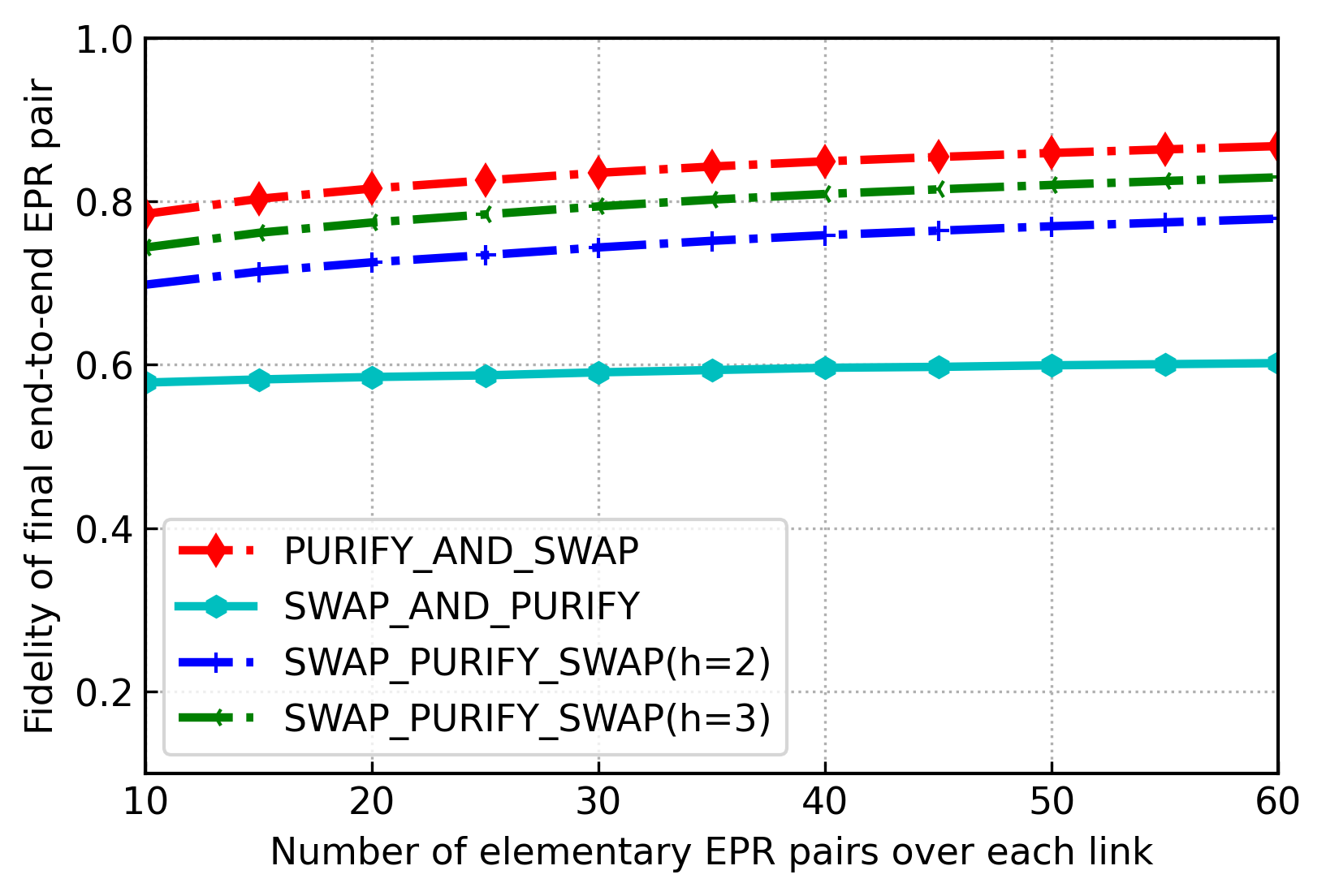}
\end{minipage}}\hspace{0.2cm}
\subfigure{
\begin{minipage}[t]{0.31\textwidth}
\includegraphics[width=\textwidth]{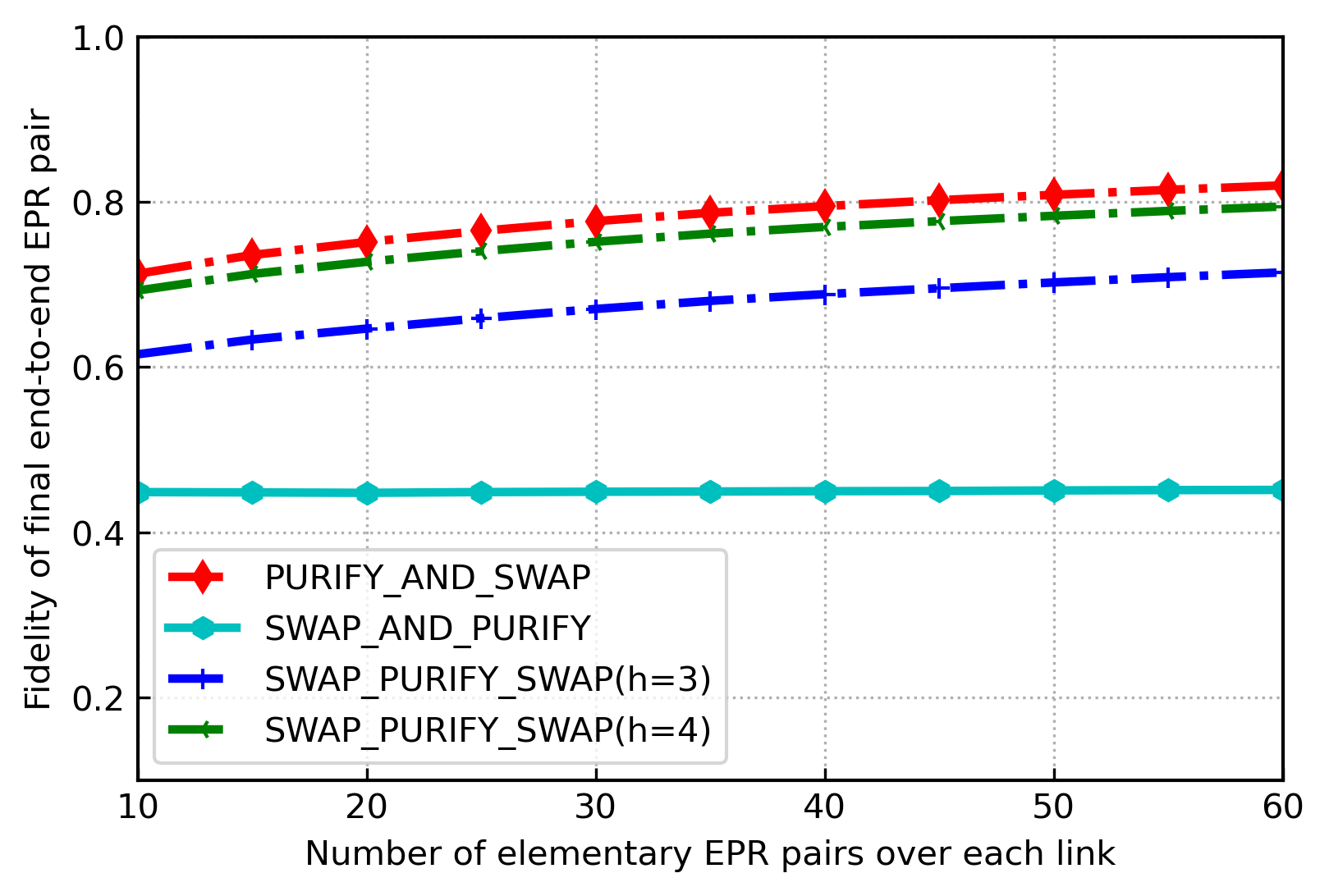}
\end{minipage}}\hspace{0.2cm}
\subfigure{
\begin{minipage}[t]{0.31\textwidth}
\includegraphics[width=\textwidth]{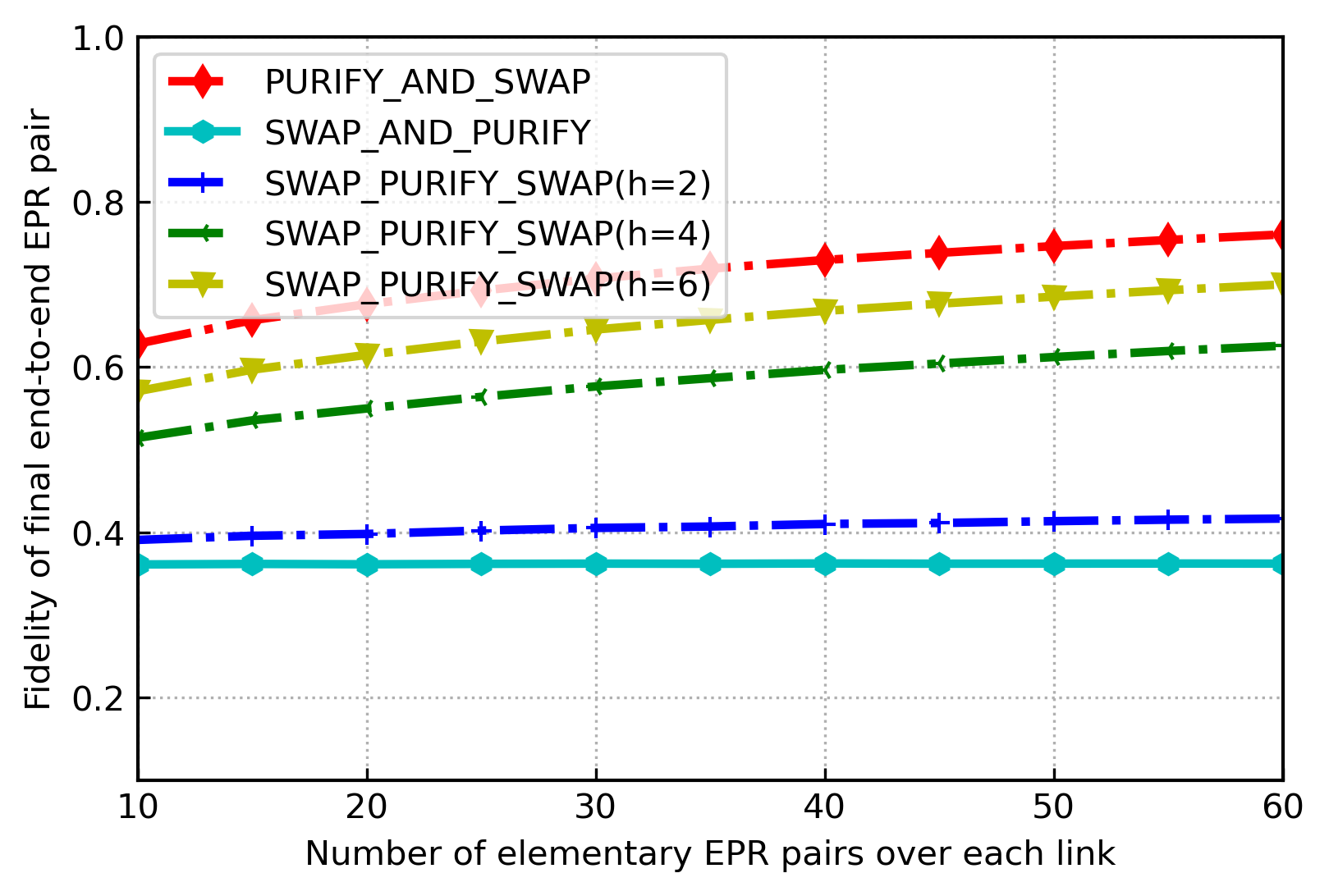}
\end{minipage}}
\vspace{-0.4cm}
\caption{Final fidelity comparison of multi-hop entanglement purification scheduling algorithms: left: $l=6$, middle: $l=9$, right: $l=12$.} 
\label{fig:simu-m-hop}
\end{figure*}

We then proceed to the multi-hop scenario to investigate the interaction between entanglement purification and entanglement swapping. Specifically, we compare \textsc{purify-and-swap} against \textsc{swap-and-purify} and another composite strategy, called \textsc{swap-purify-swap}, described as follows. 
\begin{enumerate}
    \item Divide the entire path into $h$ portions of length $\lfloor l/h\rfloor$ or $\lceil l/h\rceil$ with $l$ denoting the path length;
    \item Perform \textsc{swap-and-purify} within each portion;
    \item Stitch the $h$ purified EPR pairs into a final EPR pair.
\end{enumerate}

Figure~\ref{fig:simu-m-hop} traces the simulation results with different number of hops and elementary EPR pairs per hop, with randomly selected fidelity values in $[0.85,0.99]$. We observe that \textsc{purify-and-swap} always outperforms the other two strategies, which is consistent to our theoretical result. \textsc{swap-and-purify} performs the worst among the simulated strategies. The performance of \textsc{swap-purify-swap} depends on the parameters $h$. When $h$ approaches $l$, the performance of \textsc{swap-purify-swap} approaches that of \textsc{purify-and-swap}, with $h=l$ degenerating to \textsc{purify-and-swap}. This result demonstrates that it is usually better to perform entanglement purification over individual links, i.e., a relatively local operation, before entanglement swapping, which is a more global operation.

\subsection{Fidelity-constrained Entanglement Routing}

Our second set of simulations focus on entanglement routing, by integrating our purification scheduling algorithm.

For experiment setting, we use NetworkX to randomly generate a Waxman random graph with $300$ nodes and link capacity of $10$, and a $5\times5$ grid network with link capacity of $15$. Consistent with the experiment setting in the literature~\cite{fidelityguaranteed}, we configure the fidelity of links to follow a normal distribution within the range of $[0.8, 1]$ and randomly selected source and destination. The objective is to find the min-cost path from a random source to destination while meeting the throughput threshold of $1$ and the specified fidelity threshold. Since fidelity thresholds below $0.8$ are generally considered to lack practical relevance, we evaluate the performance of the simulated algorithms at fidelity thresholds of $0.8$, $0.85$, $0.9$ for Waxman graph (Figure~\ref{fig:single-flow-waxman-cmp}) and grid network (Figure~\ref{fig:single-grid_5x5-cmp}). 
We evaluate our algorithm with two state-of-the-art quantum routing algorithms \textsc{q-path} and \textsc{q-leap}, both designed to find end-to-end paths with fidelity and throughput guarantee~\cite{fidelityguaranteed}. To further analyze path-searching performance, we design another comparison algorithm, called \textsc{q-step}, which employs the same \textsc{pumping} purification strategy as \textsc{q-path} and \textsc{q-leap} while the same path computation strategy as our algorithm. We set the cost to be the number of qubits consumed in the corresponding path.

Figure~\ref{fig:single-flow-waxman-cmp} and~\ref{fig:single-grid_5x5-cmp} compare of our algorithm under different discretization stepsizes $\Delta\phi$ against \textsc{q-step}, \textsc{q-path} and \textsc{q-leap}. We observe from the results that our algorithm achieve the best performance in terms of success probability and path cost, which is coherent to our theoretical quasi-optimality guarantee. Furthermore, the performance gap underscores the importance of carefully optimizing the entanglement path computation by holistically integrating purification scheduling in a cross-layer manner, which is the central focus of our work.

More specifically, we observe from the results that, as the fidelity threshold increases, the success probability of the simulated algorithms decreases, and the path cost increases. This can be explained by the fact that, when the fidelity constraint becomes more stringent, we need more quantum resources to perform additional rounds of entanglement purification, leading to higher cost. If the fidelity threshold is too high, it is impossible to construct feasible end-to-end entanglement, which explains the result that the success probability is less than $1$ for some cases in the simulations. The higher success probability of \textsc{q-step} against \textsc{q-path} and \textsc{q-leap} demonstrates the performance gain brought by the path computation strategy in our algorithm, as \textsc{q-step} employs the same purification scheduling strategy as \textsc{q-path} and \textsc{q-leap}. The higher success probability of our algorithm against \textsc{q-step} further demonstrates the  performance gain brought by the purification scheduling strategy we develop. 

We conclude our simulation by comparing the running time of the simulated algorithms. We observe that our algorithm incurs higher time overhead because our approach seeks a near-optimal solution. However, this overhead can be tuned by the discretization stepsize such that, when $\Delta\phi=0.01$ and $0.02$ in the simulated networks, the running time does not exceed $5$ seconds. The typical lifespan of a qubit extends to tens of seconds. Recent advancements in the field, as highlighted by research \cite{wang2021time}, have shown that the qubit lifetime can surpass one hour. \textcolor{black}{Our algorithm allows for a trade-off between precision and efficiency by adjusting the discretization step size parameter, $\Delta\phi$. As a result, even when executed on demand, our algorithm is suitable for medium-sized quantum networks today and holds the potential for deployment in large-scale quantum networks in the foreseeable future.}

\begin{figure*}
\centering
\subfigure{
\begin{minipage}[t]{0.32\textwidth}
\includegraphics[width=\textwidth]{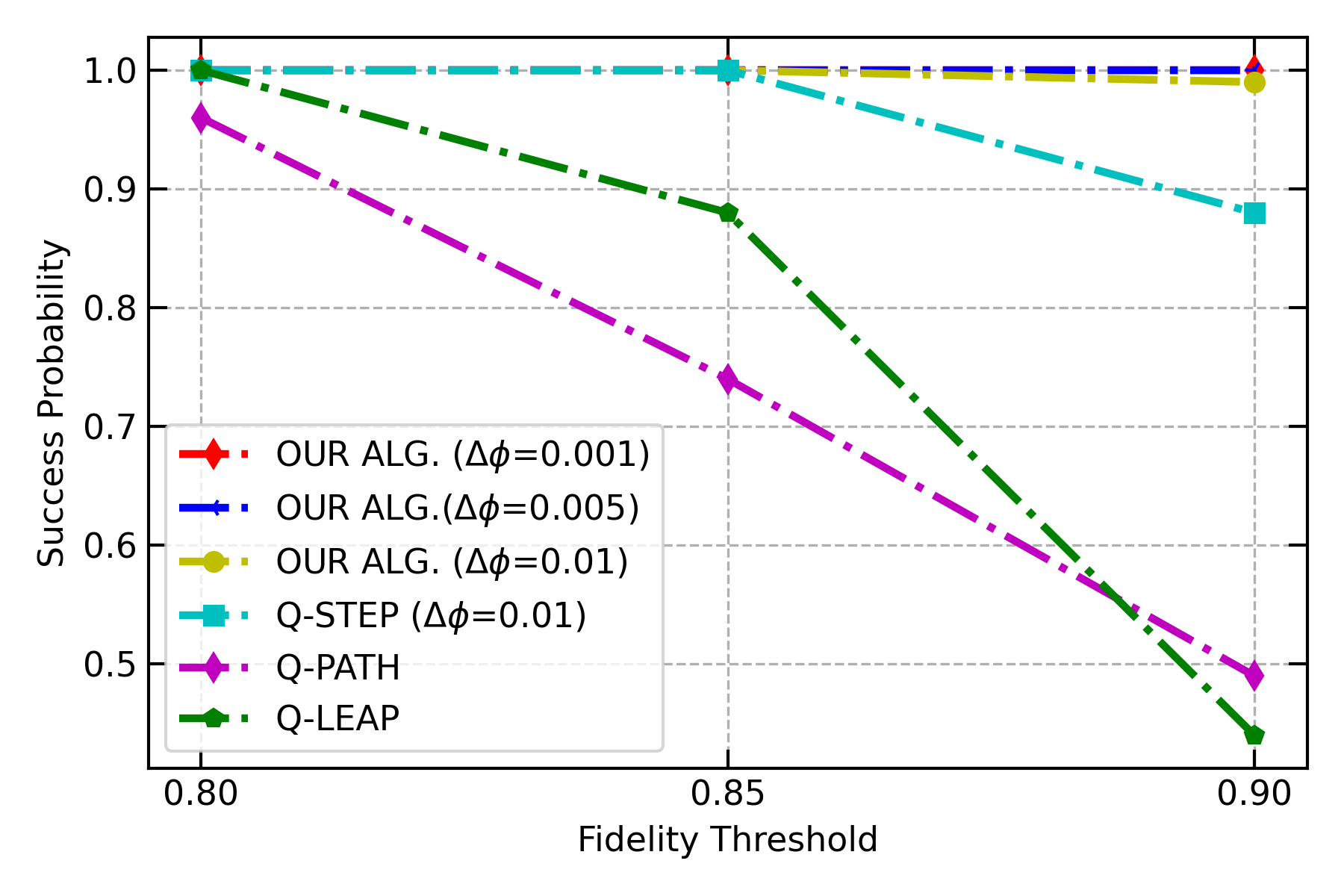}
\label{fig:waxman_probability}
\end{minipage}}
\hfill
\subfigure{
\begin{minipage}[t]{0.32\textwidth}
\includegraphics[width=\textwidth]{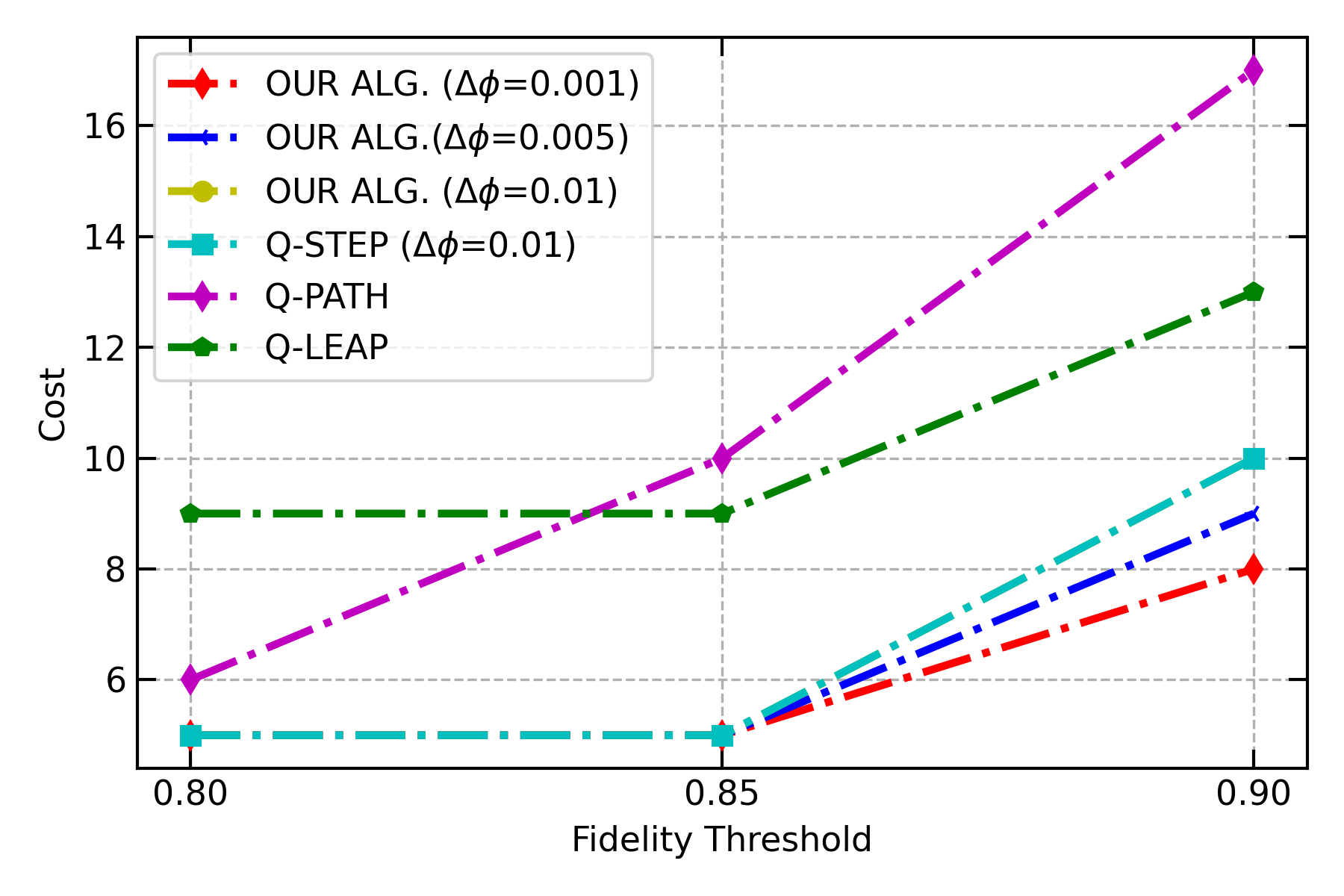}
\label{fig:waxman_cost}
\end{minipage}}
\hfill
\subfigure{
\begin{minipage}[t]{0.32\textwidth}
\includegraphics[width=\textwidth]{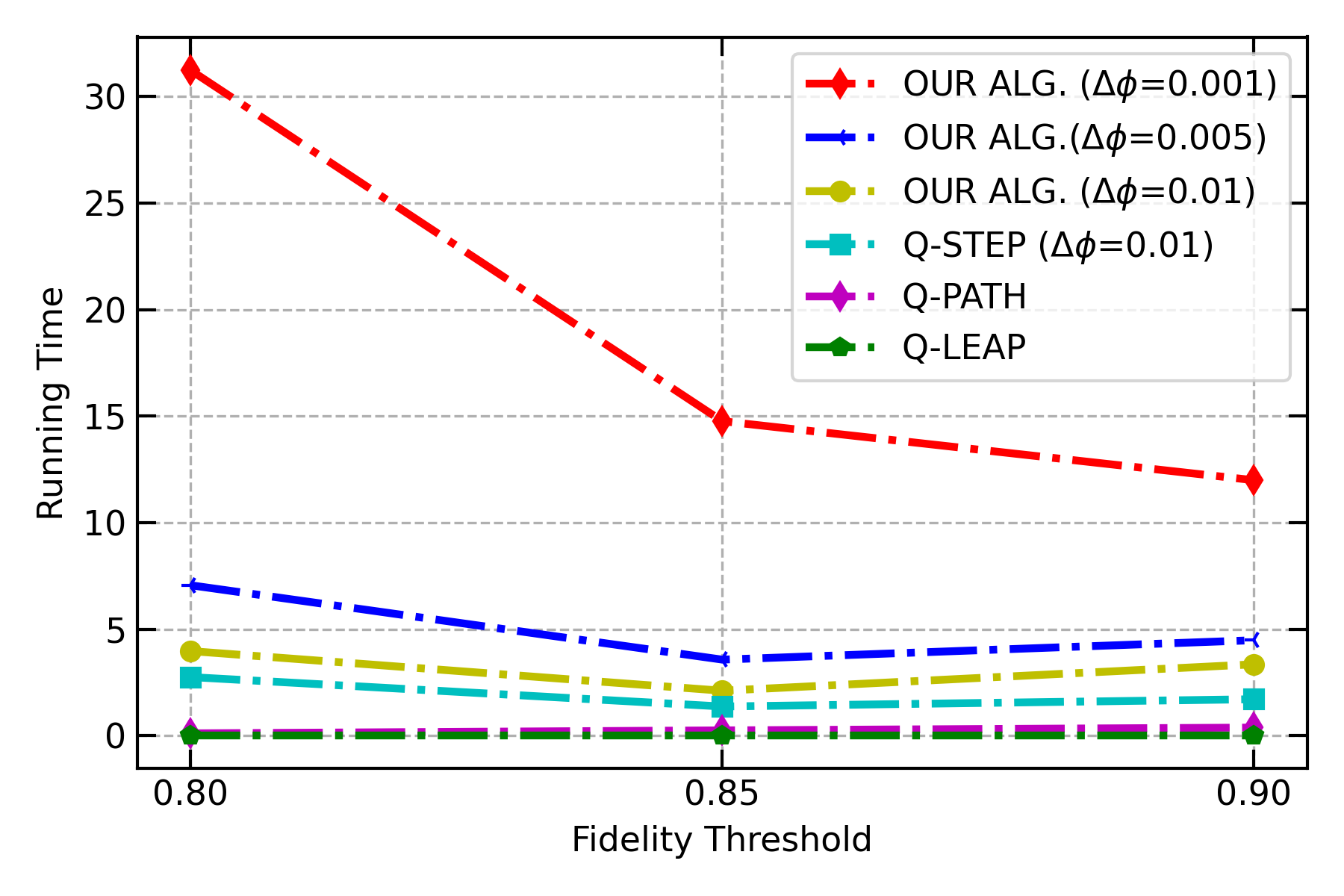}
\label{waxman_time}
\end{minipage}}
\vspace{-0.8cm}
\caption{Performance of our algorithm with \textsc{q-step}, \textsc{q-path} and \textsc{q-leap} (channel capacity = $10$, Waxman topology with $300$ nodes): left: success probability, middle: cost; right: running time} 
\label{fig:single-flow-waxman-cmp}

\subfigure{
\begin{minipage}[t]{0.32\textwidth}
\includegraphics[width=\textwidth]{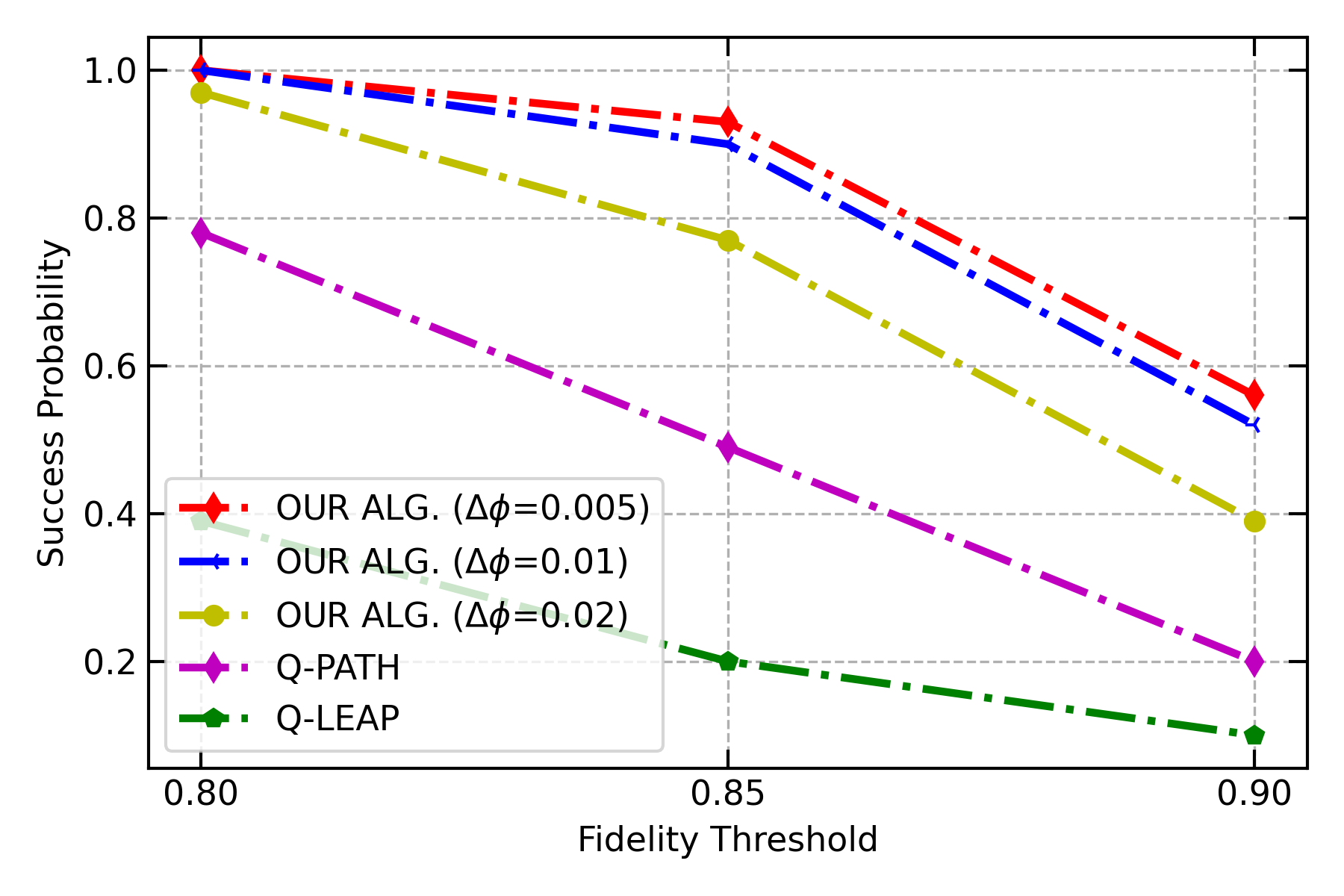}
\label{fig:grid_success}
\end{minipage}}
\hfill
\subfigure{
\begin{minipage}[t]{0.32\textwidth}
\includegraphics[width=\textwidth]{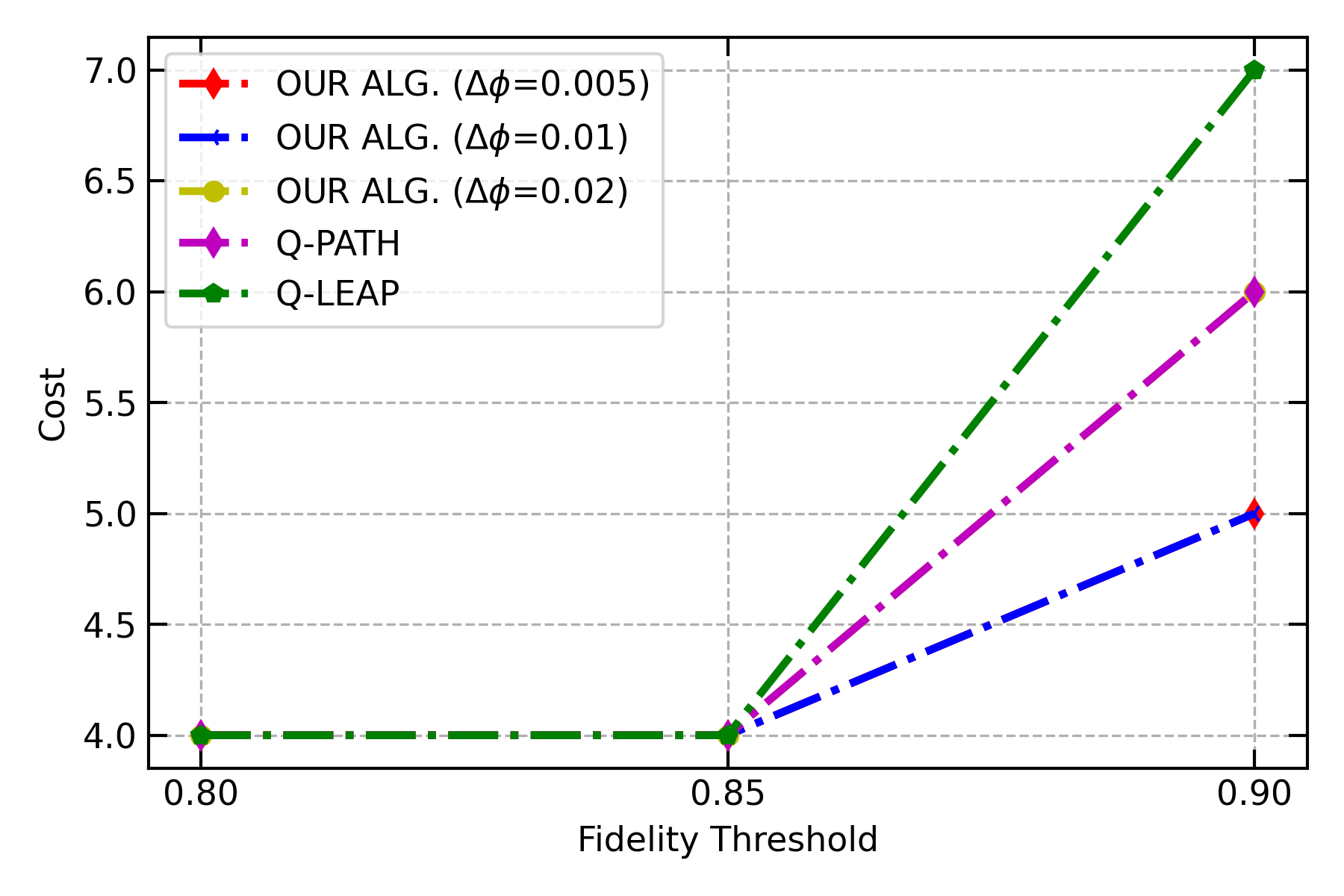}
\label{fig:grid_cost}
\end{minipage}}
\hfill
\subfigure{
\begin{minipage}[t]{0.32\textwidth}
\includegraphics[width=\textwidth]{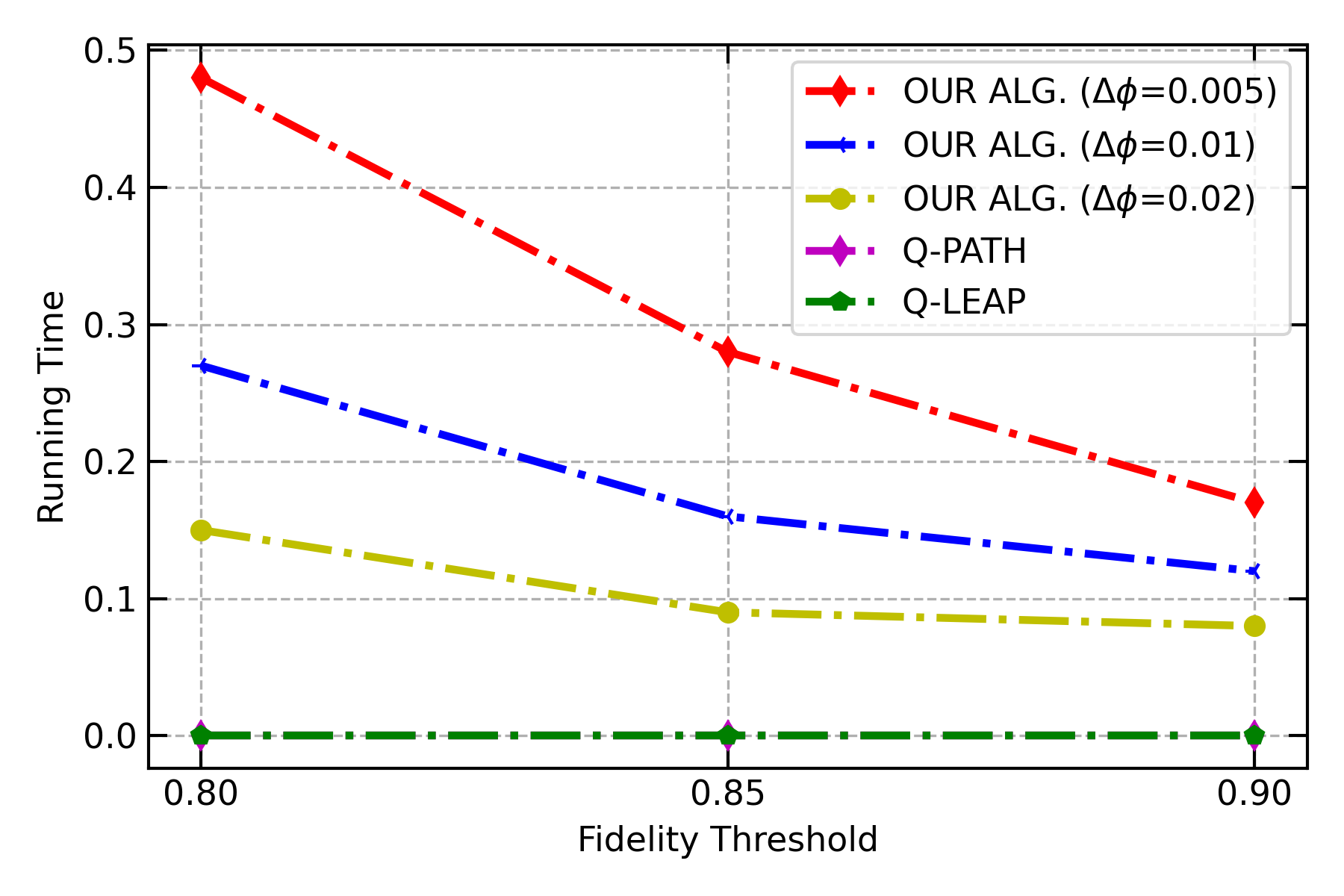}
\label{fig:grid_time}
\end{minipage}}
\vspace{-0.8cm}
\caption{Performance of our algorithm with \textsc{q-step}, \textsc{q-path} and \textsc{q-leap} (channel capacity = $15$, Grid $5\times5$): left: success probability, middle: cost; right: running time} 
\label{fig:single-grid_5x5-cmp}

\end{figure*}

\section{Related Works}
\label{sec:related-work}

Establishing end-to-end entanglements reliably and efficiently is probably one of the most fundamental tasks in quantum networks. Research on entanglement routing was initiated by the physics community, with propositions for specific network topologies such as lines, rings, and grids~\cite{Pirandola8093785,schoute2016shortcuts,9351761}. Their main focus was to develop feasible ways to achieve end-to-end entanglements to support quantum information transmission. However, relying on specific network topologies significantly limits their generality and practicality.

More recently, entanglement routing began to capture research interests from the networking community, resulting in a handful of propositions for generic network topology, e.g., ~\cite{10.1145/3387514.3405853,9488850,10150/633684,9210823}. To further ensure end-to-end fidelity, a few works have integrated entanglement purification into the process of entanglement path optimization, by solving the corresponding constrained routing problems~\cite{9796814,fidelityguaranteed,Victora2020entanglement}, among which~\cite{9796814,fidelityguaranteed} considered problem with network-layer perspective, but they only considered simple \textsc{pumping} purification scheme or enumerated all combinations.

Purification scheduling is a critical yet largely unexplored area, essential for enhancing the fidelity of final purified entanglements. A recent comprehensive survey on entanglement purification and routing~\cite{leone21} highlights the significant challenge in deriving an optimal entanglement purification schedule, labeling it as highly intractable. Due to this complexity, commonly used schemes are typically simple and intuitive algorithms, such as \textsc{pumping}~\cite{PhysRevA.59.169,briegel1998quantum} and \textsc{symmetric}~\cite{PhysRevA.54.1844,PhysRevLett.76.722,PhysRevLett.77.2818}, which lack any guarantee of optimality, as noted in recent studies~\cite{9796814,fidelityguaranteed,Victora2020entanglement,jiachenjetcas,panigrahy2022capacity}. Recent research~\cite{chenjiajsac} has considered optimal entanglement purification scheduling and routing. However, it only solve the optimal purification scheduling in the bit-flip case and does not address purification scheduling in the Werner case. As demonstrated in their paper, purification scheduling in the Werner case is more challenging due to the lack of structural properties present in the bit-flip case. \textcolor{black}{Recent work~\cite{panigrahy2022capacity,jiachenjetcas}, by experiments, discover the throughput of \textsc{purify-and-swap} outperforms \textsc{swap-and-purify}. However, They do not give any closed-from condition for the optimality of \textsc{purify-and-swap} .} This gap in research motivates our work to develop an algorithmic framework for fidelity and throughput-guaranteed entanglement path optimization with $\epsilon$-optimal purification scheduling.  

\section{Conclusion and Perspective}
\label{sec:conclu}
\textcolor{black}{In this paper, we first present an optimal entanglement purification scheduling framework for a single hop. We then analyze the optimality of \textsc{purify-and-swap} swapping strategy in multi-hop scenarios. Finally, we combine these insights to develop an polynomial time algorithmic framework to computing minimum-cost end-to-end entanglement path, subject to $\epsilon-$optimality in fidelity and throughput constraints, focusing on entanglement purification scheduling and entanglement routing. Our main technicalities hinge on the structural properties of the quantum operations, notably entanglement purification and entanglement swapping, combined with algorithmic tools in combinatorial optimization. Our algorithmic frameworks are highly flexible, allowing for a balance between accuracy and computation time through the choice of $\Delta\phi$. Our algorithmic framework can also serve as a key functional blocks for entanglement configuration, enabling high-fidelity quantum information transmission by effectively leveraging limited quantum resources, thus fully harnessing the unmatched potential of quantum networks. The superiority of our algorithm frameworks are also demonstrated by extensive simulations.} 

Looking ahead, we plan to extend our framework by incorporating measurement error to account for imperfect measurements. One direction for future research is to factor in decoherence time, further enhancing the realism of the model. Another avenue is to allow EPR pairs across different entanglement paths to be purified, and to design a joint purification scheduling and path optimization algorithm in this more complex context. 


\appendix

\section{Proofs of theorems and lemmas}

\subsection{Proof of Theorem~\ref{thm:opt-single-hop}}

    Let $T^*$ denote the optimal purification tree, $b^*$ denote the number of leaves in $T^*$, $f^*$ denote the fidelity of its root, $\xi^*$ denote the expected number of purified EPR pairs mapping to its root.  Define $l^*\triangleq(b^*,\hat{f^*},\hat{\xi}^*,T^*)$, where $\hat{\xi}^*\triangleq \left\lceil\frac{\xi^*}{\Delta\xi}\right\rceil\Delta\xi$ and $\hat{f}^*\triangleq \left\lceil\frac{f^*}{\Delta f}\right\rceil\Delta f$ denote the discretized value of $\xi^*$ and $f^*$,respectively. 
    We can prove by induction that $l^*$ is dominated by $l'$ by noticing that, if a pair of purification tree $T_1$ and $T_2$ are dominated by the purification trees  $T_1'$ and $T_2'$ in $\cal L$, then it holds that the purification tree whose children are $T_1$ and $T_2$ is also dominated by the purification tree whose children are $T_1'$ and $T_2'$. Therefore, by induction we can prove 
    \begin{equation}
    b'\le b^*, \hat{f'} \ge \hat{f}^* \ge f^*, \ \hat{\xi}'\ge \hat{\xi}^*\ge \xi^*.
    \label{eq:opt-appr}
    \end{equation}
    
    {Noticing~\eqref{eq:opt-appr}, it suffices to prove $\xi'\ge (1-\epsilon)\xi^*$ and $f'\ge (1-\epsilon)f^*$ by proving $\xi'\ge (1-\epsilon)\hat{\xi}'$ and $f'\ge (1-\epsilon)\hat{f}'$, respectively}. Noticing the process of construction of purification trees from two sub-trees, we first prove the following auxiliary lemma.
    
    \begin{lemma}
        Given a pair of entries $l_i\triangleq (b_i,f_i,\hat{\xi}_i,T_i)$, $i=\{1,2\}$, let $l_3\triangleq(b_3\triangleq b_1+b_2,\hat{f_3},\hat{\xi}_3,T_3)$ denote the entry by merging $T_1$ and $T_2$, i.e., $T_3$ has $T_1$ and $T_2$ as children. Under the condition $\Delta\xi\le \min\{b_1,b_2\}\xi_3\epsilon$ and $\Delta f\le \min\{b_1,b_2\}f_3\epsilon$, the following inequality holds:
        \begin{align*}
            (1-b_3\epsilon)\hat{\xi}_3 \le \xi_3 \le (1+b_3\epsilon)\hat{\xi}_3. \\
            (1-b_3\epsilon)\hat{f}_3 \le f_3 \le (1+b_3\epsilon)\hat{f}_3.
        \end{align*}

        \label{lemma:aux-app}
    \end{lemma}

    \begin{proof}
        We prove the condition $\Delta \xi$ firstly. 
        We prove the lemma by induction on $b_1$ and $b_2$. When $b_1=b_2=1$, we have         $$\xi_3=P(f_e,f_e), \ \hat{\xi}_3=\left\lceil\frac{P(f_e,f_e)}{\Delta\xi}\right\rceil\Delta\xi.$$
        It holds straightforwardly that $\xi_3\le (1+b_3\epsilon)\hat{\xi}_3$. We next prove $\xi_3\ge(1-b_3\epsilon)\hat{\xi}_3$. Noticing that 
        $$\left\lceil\frac{P(f_e,f_e)}{\Delta\xi}\right\rceil<\frac{P(f_e,f_e)}{\Delta\xi}+1,$$
        we have 
        \begin{align*}
            \xi_3-(1-b_3\epsilon)\hat{\xi}_3 &> P(f_e,f_e)- (1-b_3\epsilon)\left(   P(f_e,f_e)+\Delta\xi\right) \\
            &=b_3\epsilon P(f_e,f_e)-(1-b_3\epsilon)\Delta\xi \\
            &=\xi_3b_3\epsilon-(1-b_3\epsilon)\Delta\xi> 0.
        \end{align*}
        The last inequality follows from the condition $\Delta\xi\le b_3\xi_3\epsilon$.
        We next prove the induction case. We have $$\xi_3=P(f_1,f_2)\xi_1, \ \hat{\xi}_3=\left\lceil\frac{P(f_1,f_2)\hat{\xi}_1}{\Delta\xi}\right\rceil\Delta\xi.$$
        We first prove $\xi_3\le(1+b_3\epsilon)\hat{\xi}_3$. By induction we have $\xi_1\le(1+b_1\epsilon)\hat{\xi}_1$. 
        Then, we have that
        \begin{align*}
            (1+b_3\epsilon)\hat{\xi}_3 &> (1+b_3\epsilon)\left\lceil\frac{P(f_1,f_2)\xi_1}{\Delta\xi}\right\rceil\Delta\xi \\
            &> (1+b_3\epsilon)P(f_1,f_2)\xi_1 > \xi_3.
        \end{align*}
        We then prove $\xi_3\ge(1-b_3\epsilon)\hat{\xi}_3$. By induction we have 
        $$\left\lceil\frac{P(f_1,f_2)\hat{\xi}_1}{\Delta\xi}\right\rceil<\frac{P(f_1,f_2)\hat{\xi}_1}{\Delta\xi}+1\le \frac{P(f_1,f_2)\xi_1}{(1-b_1\epsilon)\Delta\xi}+1.$$
        We then have 
        {\small
        \begin{align*}
            \xi_3-(1-b_3\epsilon)\hat{\xi}_3 &> P(f_1,f_2)\xi_1- 
            (1-b_3\epsilon)\left(   \frac{P(f_1,f_2)\xi_1}{(1-b_1\epsilon)}+\Delta\xi\right) \\
            &= \frac{P(f_1,f_2)\xi_1b_1\epsilon}{(1-b_1\epsilon)}-(1-b_3\epsilon)\Delta\xi>0.
        \end{align*}}
        In the above proof, the equality follows from $b_3=2b_1$ and the last inequality follows from the condition regarding $\epsilon$.
        The prove of condition $\Delta f$ is basically same.
    \end{proof}
    It follows from Lemma~\ref{lemma:aux-app} that, under the condition $\Delta\xi\le \min\{b_1,b_2\}\xi_3\epsilon$ and $\Delta f\le \min\{b_1,b_2\}f_3\epsilon$, it holds that $\xi_3\ge(1-b_3\epsilon)\hat{\xi}_3$ and $f_3\ge(1-b_3\epsilon)\hat{f}_3$.
    By applying the result recursively from the leaves to the root of $T'$, we can prove that, if $\Delta\xi\le b'\xi'\epsilon$ and $\Delta f\le b'f'\epsilon$, it holds that $\xi'\ge (1-b'\epsilon-o(b'\epsilon))\hat{\xi}'$ and $f'\ge (1-b'\epsilon-o(b'\epsilon))\hat{f}'$. It then holds asymptotically that 
    \begin{align*}
        \lambda(T')=n\xi'\ge (1-\epsilon)n\xi^*=(1-\epsilon)\lambda(T^*). \\
        f(T')=nf'\ge (1-\epsilon)nf^*=(1-\epsilon)f(T^*).
    \end{align*}
    Our algorithm thus outputs an $\epsilon-$optimal solution of Problem~\ref{pb:purif-single-hop}. The proof is completed.

\subsection{Proof of Lemma~\ref{fact:aux}}

Let $a,b,c,d$ denote $f_i$, $i=1,2,3,4$, respectively. Proving the result of the lemma regarding fidelity is equivalent to show that
\begin{equation}
    \Delta\triangleq f(F(a,b),F(c,d))-F(f(a,c),f(b,d))>0,
    \label{eq:ob}
\end{equation}
under the following condition 
\begin{equation}
    0.5\le a\le b\le1, \ 0.7\le c\le d\le1.
    \label{eq:condition}
\end{equation}

Algebraically, we arrange $\Delta$ as 
$$\Delta\triangleq \frac{2}{3}\frac{\Delta_1}{\Delta_2},$$
where the nominator $\Delta_1$ is a polynomial of degree $\le2$ in $a,b,c,d$, the denominator $\Delta_2$ is 
\begin{multline}
    \Delta_2 = \big(2a(4b-1) - 2b + 5\big)\big(2c(4d-1) - 2d + 5\big) \\
    \times \big(2(a(4b-1) - b)(4c-1)(4d-1) + 8cd - 2(c + d) + 41\big).
\end{multline}

As each of factors in the latter expression is affine in $a,b,c,d$, it is straightforward to check that $\Delta_2>0$ given~\eqref{eq:condition}. Therefore, it suffices to prove $\Delta_1>0$.

Recall that $\Delta_1$ is a polynomial of degree $\le2$ in $d$. The coefficient of $d^2$ in $\Delta_1$ is $2(4b-1)(4c-1)\Delta_d$, where $\Delta_d$ is a polynomial in $a,b,c$ of degree $\le1$, and the coefficient of $c$ is $16(4a-1)(-4+a+b+2ab)<0$ given~\eqref{eq:condition}. Therefore, we have $\Delta_2\le \Delta_2|_{c=0.7}$. 

Algebraically, observe that $\Delta_2|_{c=0.7}$ is now a polynomial in $a,b$ of degree $\le1$, with the coefficient of $b$ being $35(4a-1)(64a-13)>0$ given~\eqref{eq:condition}. Hence, we have
$$\Delta_2\le \Delta_2|_{c=0.7,b=a}=-0.6(1-a)(7+178a+256a^2)<0.$$
Therefore, the coefficient of $d^2$ in $\Delta_1$ is negative. Hence, $\Delta_1$ is concave in $d$. It is then enough to prove $\Delta_1>0$ for $d\in \{1,c\}$.

This is achieved using similar concavity-cum-monotonicity reasoning with factoring polynomials, by noticing that (1) $\frac{\Delta_1|_{d=1}}{9(1-c)}$ is a polynomial of degree $\le1$ in $c$, and (2) $\frac{\Delta_1|_{d=c}}{9(1-c)(4c-1)}$ is a polynomial of degree $\le2$ in $b$, which is concave in $b$. The superiority of \textsc{purify-and-swap} in terms of fidelity is thus proved.

We proceed to prove the superiority of \textsc{purify-and-swap} in terms of success probability. Let $p_s$ denote the success probability of entanglement swapping at Charlie. We can derive the success probability of \textsc{purify-and-swap} and \textsc{swap-and-purify} as $P(a,b)P(c,d)p_s$ and $P(f(a,c),f(b,d))p_s^2$, respectively. We then have 

\begin{multline*}
    P(a,b)P(c,d) - P(f(a,c),f(b,d))p_s \ge \\
    \left[\frac{8}{9}ab - \frac{2}{9}(a+b) + \frac{5}{9}\right] 
    \left[\frac{8}{9}cd - \frac{2}{9}(c+d) + \frac{5}{9}\right] \\
    - 0.818\left[\frac{8}{9}xy - \frac{2}{9}(x+y) + \frac{5}{9}\right],
\end{multline*}
where 
\begin{align*}
    x&=f(a,c)=\frac{1}{4}\left[1+\frac{1}{3}(4a-1)(4c-1)\right]    \\
    y&=f(b,d)=\frac{1}{4}\left[1+\frac{1}{3}(4b-1)(4d-1)\right]
\end{align*}

With tedious while straightforward algebraic operations on derivatives, we can prove that $P(a,b)P(c,d)-P(f(a,c),f(b,d))$ is monotonously increasing in $a$, $b$, $c$ and $d$. Noticing $0.818\le p_s\le 1$, we have 
\begin{multline*}
    P(a,b)P(c,d)p_s - P(f(a,c),f(b,d))p_s^2 \ge \\  
    p_s \big[P(a,b)P(c,d) 
     - 0.818 \cdot P(f(a,c),f(b,d))\big] \ge \\
     p_s \big[(P(0.7,0.7))^2- 0.818 \cdot P(f(0.7,0.7), f(0.7,0.7))\big] \\
    = 0.0004p_s > 0.
\end{multline*}

\subsection{Proof of Theorem~\ref{theorem:opt-purif-schedule}}

We prove the fidelity result of the theorem. The result concerning success probability can be proved similarly with straightforward algebraic operations and is thus omitted for brevity. Assume, by contradiction, that \textsc{purify-and-swap} produces a sub-optimal end-to-end EPR pair. Let $\pi^*$ denote an optimal policy. It follows from the assumption that there must exist an entanglement purification operation between a pair of non-elementary EPR pairs. Let $\lambda$ denote the purified EPR pair and let $\lambda_1$ and $\lambda_2$ denote the two non-elementary EPR pair producing $\lambda$. Let $v_1$ and $v_2$ denote the two end-points of the $\lambda$. They are also the end-points of $\lambda_1$ and $\lambda_2$. Let $v_0$ denote the node preceding $v_2$ in $\lambda$, as shown in Figure~\ref{fig:proof}. Now we construct another policy $\pi'$ from $\pi^*$. Instead of forming $\lambda_1$ and $\lambda_2$ and then purifying them to $\lambda$, we first perform purification on the two EPR pairs between $v_1$ and $v_0$, denoted by $\mu_1$ and $\mu_2$, and the two EPR pairs between $v_0$ and $v_2$, denoted by $\nu_1$ and $\nu_2$, and then perform entanglement swapping by stitching the purified EPR pair to produce a final EPR pair between $v_1$ and $v_2$. The rest part of $\pi'$ is identical to $\pi^*$. The following properties follow from the optimality of $\pi^*$. 
    \begin{itemize}
        \item $\lambda\ge\lambda_i\ge 0.5$ for $i=1,2$. This follows from Lemma~\ref{fact:aux}, because if at least one of $\lambda_1$ and $\lambda_2$ is smaller than $0.5$, then we are better off by not performing the entanglement purification operation between them, thus contradicting with the optimality of $\pi^*$.
        \item $\mu_i\ge \lambda_i$. This follows directly from the property of the entanglement swapping operation.
        \item $\lambda\ge\lambda'$. This follows from the optimality of $\pi^*$.
    \end{itemize}

\begin{figure}[!ht]
\centering
\begin{tikzpicture}[line width=0.75pt]
\draw (0,0) circle (0.3);
\fill (0.1,0.1) circle (0.06);
\fill (0.1,-0.1) circle (0.06);
\draw (3,0) circle (0.3);
\fill (2.9,0.1) circle (0.06);
\fill (2.9,-0.1) circle (0.06);
\draw (0.1,0.1) -- (2.9,0.1);
\draw (0.1,-0.1) -- (2.9,-0.1);
\draw [line width=1pt,->] (1.5,-0.8) -- (1.5,-2);
\node[right] at (1.5,-1.3) {purify};
\draw (0,-2.5) circle (0.3);
\fill (0.1,-2.5) circle (0.06);
\draw (3,-2.5) circle (0.3);
\fill (2.9,-2.5) circle (0.06);
\draw (0.1,-2.5) -- (2.9,-2.5);
\node[right] at (1.2,0.3) {$\lambda_1$};
\node[right] at (1.2,-0.3) {$\lambda_2$};
\node[right] at (1.2,-2.7) {$\lambda$};

\draw (4.5,0) circle (0.3);
\fill (4.6,0.1) circle (0.06);
\fill (4.6,-0.1) circle (0.06);
\draw (6,0) circle (0.3);
\fill (6.1,0.1) circle (0.06);
\fill (6.1,-0.1) circle (0.06);
\fill (5.9,0.1) circle (0.06);
\fill (5.9,-0.1) circle (0.06);
\draw (7.5,0) circle (0.3);
\fill (7.4,0.1) circle (0.06);
\fill (7.4,-0.1) circle (0.06);
\draw (4.6,0.1) -- (5.9,0.1);
\draw (4.6,-0.1) -- (5.9,-0.1);
\draw (6.1,0.1) -- (7.4,0.1);
\draw (6.1,-0.1) -- (7.4,-0.1);
\draw [line width=1pt,->] (5.2,-0.6) -- (5.2,-1.1);
\node[right] at (5.2,-0.7) {purify};
\draw [line width=1pt,->] (6.7,-0.6) -- (6.7,-1.1);
\node[right] at (6.7,-0.7) {purify};
\draw (4.5,-1.3) circle (0.3);
\fill (4.6,-1.3) circle (0.06);
\draw (6,-1.3) circle (0.3);
\fill (5.9,-1.3) circle (0.06);
\fill (6.1,-1.3) circle (0.06);
\draw (7.5,-1.3) circle (0.3);
\fill (7.4,-1.3) circle (0.06);
\draw (4.6,-1.3) -- (5.9,-1.3);
\draw (6.1,-1.3) -- (7.4,-1.3);
\draw [line width=1pt,->] (6,-1.8) -- (6,-2.3);
\node[right] at (6,-2) {swap};
\draw (4.5,-2.5) circle (0.3);
\fill (4.6,-2.5) circle (0.06);
\draw (7.5,-2.5) circle (0.3);
\fill (7.4,-2.5) circle (0.06);
\draw (4.6,-2.5) -- (7.4,-2.5);

\node[right] at (5,0.3) {$\mu_1$};
\node[right] at (5,-0.3) {$\mu_2$};
\node[right] at (6.5,0.3) {$\nu_1$};
\node[right] at (6.5,-0.3) {$\nu_2$};
\node[right] at (5,-1.5) {$\mu$};
\node[right] at (6.5,-1.5) {$\nu$};
\node[right] at (5.8,-2.7) {$\lambda'$};

\node[right] at (-0.2,-3.1) {$v_1$};
\node[right] at (2.8,-3.1) {$v_2$};
\node[right] at (4.3,-3.1) {$v_1$};
\node[right] at (7.3,-3.1) {$v_2$};
\node[right] at (5.8,-3.1) {$v_0$};
\end{tikzpicture}
\caption{Illustration of the proof of Theorem~\ref{theorem:opt-purif-schedule}.}
\label{fig:proof}
\end{figure}
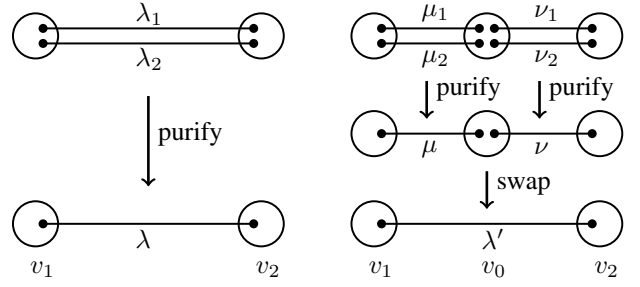
    On the other hand, noticing that $\nu_i$, $i=1,2$, are elementary EPR pairs, we have {$\nu_i=f\ge 0.7$, $i=1,2$}. It then follows from Lemma~\ref{fact:aux} that $\lambda<\lambda'$, by mapping $\mu_1$ and $\mu_2$ to $f_1$ and $f_2$, $\nu_1$ and $\nu_2$ to $f_3$ and $f_4$, which contradicts to $\lambda\ge\lambda'$, noticing the optimality of $\pi^*$. {According to $\mu_i\ge \lambda_i$, we have $\mu_i\ge 0.7$ by swapping operation. The path are constructed by these `blocks' (Fig.~\ref{fig:proof}) mentioned above.} The theorem is thus proved. 

\subsection{Proof of Theorem~\ref{theorem:opt-path}}

    Let $P^*$ and $\hat{P}$ denote the optimal path  and the path output by Algorithm~\ref{alg:mincost}. We need to prove 
    $$C(\hat{P})\le C(P^*), \ \psi(\hat{P})\ge(1-\epsilon)\psi_0, \ \text{and} \  \phi(\hat{P})\ge (1-\epsilon)\phi_0.$$
    
    Let $P_0=\{s_{Q_s},v^1_{i_1},\cdots,v^{r}_{i_r}\}$ denote the path in $G'$ corresponding to $P^*$, where $v^r=t$ and $r$ denotes the length of $P_0$. It follows from the algorithm that the entry corresponding to the sub-path of $P_0$ from $s_{Q_s}$ to $v^k_{i_k}$, $1\le k\le r$, is added in ${\cal L}_{v^k_{i_k}}$ if it is not dominated by any other entry. Therefore, it follows from Algorithm~\ref{alg:mincost} that $C(\hat{P})\le C(P^*)$.
    
    On the other hand, the total error introduced by discretization along $P_0$ is upper-bounded by $r\cdot\Delta\phi$ when computing $\phi(\hat{P})$, which is then upper-bounded by $|V|\cdot\Delta\phi$ as $r\le |V|$. 
    We then have 
    $$\phi(\hat{P})\ge \phi(P^*)-|V|\cdot\Delta\phi.$$
    It then follows from $\phi(P^*)\ge\phi_0$ that, if {$\Delta\phi\le \epsilon\phi_0/|V|$}, $\phi(\hat{P})\ge (1-\epsilon)\phi_0$.
    
    Symmetrically, we can prove that, if {$\Delta\psi\le \epsilon \psi_0/|V|$}, $\psi(\hat{P})\ge(1-\epsilon)\psi_0$. 
    The theorem is thus proved.

\subsection{Proof of Lemma~\ref{lemma:aux3}}

We define a set of random variables:
$$Y_{kv}\triangleq \sum_{v\in P_{ki}} a_{kiv}\hat{x}_{ki} \quad k\in{\cal K},$$
where $P_{ki}$ denotes path $i\in{\cal R}_k$. It holds that 
$$\mathbb{E}[Y_{kv}]=\sum_{v\in P_{ki}} a_{kiv}x_{ki}^*.$$
We give an auxiliary lemma trivially provable by induction.
    \begin{lemma}
        Given any set $\cal B$ of non-negative real numbers $\{b_i\}$, the following inequality holds:
        $$1+\sum_{{\cal B}} b_i\le \prod_{{\cal B}} (1+b_i).$$
        \label{lemma:aux4}
    \end{lemma}

    We now turn to prove Lemma~\ref{lemma:aux3}. For any $\gamma<0$, we have
    {\small
    \begin{multline*}
        \Pr\left[\sum_{k\in{\cal K}}\sum_{v\in P_{ki}} a_{kiv}\hat{x}_{ki}>(1+\delta)\beta Q_v\right] \\ 
        = \Pr\left[\sum_{k\in{\cal K}}Y_{kv}>(1+\delta)\beta Q_v\right]  \\
        = \Pr\left[\exp \left(\gamma\sum_{k\in{\cal K}}Y_{kv}\right)>\exp (\gamma(1+\delta)\beta Q_v)\right] \\
        < \exp (-\gamma(1+\delta)\beta Q_v)\cdot \mathbb{E}\left[\exp\left(\gamma\sum_{k\in{\cal K}}Y_{kv}\right)\right],
    \end{multline*}
    }
    where the last inequality follows from Markov inequality.
    
    Recall the definition of $Y_{kv}$, we have
    \begin{align*}
        \mathbb{E}\left[e^{\gamma Y_{kv}}\right] &= \sum_{v\in P_{ki}} x_{ki}^*\cdot e^{\gamma a_{kiv}}+ 1- \sum_{v\in P_{ki}} x_{ki}^* \\
        & \le \!\!\! \prod_{v\in P_{ki}} \!\!\! \left(x_{ki}^* e^{\gamma a_{kiv}}+ 1- x_{ki}^*\right) \\ 
        &= \!\!\! \prod_{v\in P_{ki}} \!\!\! \left(x_{ki}^* \left(e^{\gamma a_{kiv}}-1\right)+ 1\right) \\
        &< \!\!\!\prod_{v\in P_{ki}} \!\!\! \exp \left(x_{ki}^* \left(e^{\gamma a_{kiv}}-1\right)\right) \\
        &=  \exp \!\!\!\sum_{v\in P_{ki}} \!\!\! \left(x_{ki}^* \left(e^{\gamma a_{kiv}}-1\right)\right)
    \end{align*}
    , where the first inequality follows from Lemma~\ref{lemma:aux4} by setting $b_i\triangleq x_{ki}^* \left(e^{\gamma a_{kiv}}-1\right)$. The second
    inequality follows from $e^x>1+x$, $\forall x>0$.
    
    Posing $\gamma=\ln(1+\delta)$, we have
    \begin{multline*}
    \exp \left( \sum_{v \in P_{ki}} x_{ki}^* \left( e^{\gamma a_{kiv}} - 1 \right) \right) \\
    = \exp \left( \sum_{v \in P_{ki}} x_{ki}^* \left( (1+\delta)^{a_{kiv}} - 1 \right) \right) \\
    < \exp \left( \sum_{v \in P_{ki}} x_{ki}^* a_{kiv} \delta \right).
    \end{multline*}
    
    We hence get
    \begin{multline*}
        \Pr\left[\sum_{k\in{\cal K}}\sum_{v\in P_{ki}} a_{kiv}\hat{x}_{ki}>(1+\delta)\beta Q_v\right] \\
        < \exp(-\gamma(1+\delta)\beta Q_v)\cdot \exp \sum_{k\in {\cal K}}\sum_{v\in P_{ki}} x_{ki}^*a_{kiv}\delta \\
        < \exp (\beta Q_v \delta-\gamma(1+\delta)\beta Q_v) < \left(\frac{e^\delta}{(1+\delta)^{(1+\delta)}}\right)^{\beta Q_v},
    \end{multline*}
    where the second inequality follows from the constraint in the LP relaxation $\sum_{k\in {\cal K}}\sum_{v\in P_{ki}} x_{ki}^*a_{kiv}\le \beta Q_v$.
    The lemma is thus proved.

\subsection{Proof of Theorem~\ref{theorem:mult-flow}}

Let $\mathbf{w}$ denote the vector of $w_k$, $k\in{\cal K}$. We first prove that $\mathbf{\hat{x}}$ violates none of the constraints of the original ILP with prob. $\ge1/3$. Denote ${\cal E}_v$ the event that the constraint regarding $v$ is violated. Recall Lemma~\ref{lemma:aux3}, we have
    \begin{align*}
        \Pr[{\cal E}_v]&=\Pr\left[\sum_{k\in{\cal K}}\sum_{v\in P_{ki}} a_{kiv}\hat{x}_{ki}> Q_v\right] \\
        &< \Pr\left[\sum_{k\in{\cal K}}\sum_{v\in P_{ki}} a_{kiv}\hat{x}_{ki}>(1+\epsilon)(1-\epsilon) Q_v\right] \\
        &<\left(\frac{e^\epsilon}{(1+\epsilon)^{(1+\epsilon)}}\right)^{(1-\epsilon)Q_v} \\
        &<\exp \left((1-\epsilon)Q_v\left(\epsilon-(1+\epsilon)\ln(1+\epsilon)\right)\right) \\
        &< \exp(-(1-\epsilon)\epsilon^2Q_v).
    \end{align*}
    It then follows from the concavity of $e^{-x}$ that 
    \begin{align*}
        \Pr\left[\bigcup_{v\in V} {\cal E}_v\right] &\le \sum_{v\in V} \Pr[{\cal E}_v] < \sum_{v\in V} \exp((1-\epsilon)\epsilon^2Q_v)<\frac{1}{3},
    \end{align*}
    where the last inequality follows from the condition~\eqref{eq:cond-mult-flow}.
    
    On the other hand, by observing the LP relaxation, we have $$\mathbb{E}[\mathbf{w\hat{x}}]\ge \beta \mathbf{wx^*} = (1-\epsilon) \mathbf{wx^*}.$$ 
    We then have 
    \begin{align*}
        \Pr[\mathbf{w\hat{x}}\ge (1-2\epsilon) \mathbf{wx^*}] &= 1-\Pr[\mathbf{w\hat{x}}<(1-2\epsilon) \mathbf{wx^*}] \\
        &> 1-\Pr[\mathbf{w\hat{x}}<(1-\epsilon) \mathbb{E}[\mathbf{w\hat{x}}]] \\
        &> 1-\exp (-0.5\epsilon^2\mathbb{E}[\mathbf{w\hat{x}}]) \\
        &> 1-\exp (-0.5\epsilon^2(1-\epsilon)\mathbf{wx^*}) \\ 
        &> {2}/{3}
    \end{align*}
    , where the second inequality follows from Chernoff bound, the last inequality follows from the condition~\eqref{eq:cond-mult-flow} by noticing that $\mathbf{wx^*}\ge \min_k w_k$.\footnote{We assume that there exists at least a non-zero feasible solution in the LP.}  Hence, we have 
    \begin{multline*}
        \Pr[\text{$\mathbf{\hat{x}}$ is a $2\epsilon-$optimal solution of original ILP}] \\
        >\Pr[\mathbf{w\hat{x}}\ge (1-\epsilon) \mathbf{wx^*}]-\Pr\left[\bigcup_{v\in V(G)} {\cal E}_v\right] >\frac{1}{3}.
    \end{multline*}
    The theorem is thus proved.
    
\bibliographystyle{abbrv}
\bibliography{reference}  

\end{document}